%% file: main.tex
\newif \ifproofs
\proofstrue

\newif \ifnotes % toggle comments and TODOs 
\documentclass[11pt]{article}
\usepackage{amssymb}
\usepackage{mathrsfs,amsmath}
\usepackage{graphicx}
\usepackage{amsmath}
\usepackage{amsthm}
\usepackage{bbm}
\usepackage{dsfont}
\usepackage{listing}
\usepackage{hyperref}
\usepackage{color}
\usepackage[ruled, vlined,linesnumbered, noend]{algorithm2e}
\usepackage[htt]{hyphenat} % enable texttt to hyphenate and not run off the margins. 
\newtheorem{observation}{Observation}

\newtheorem{assumption}{Assumption}
\theoremstyle{definition}
\newtheorem{definition}{Definition}
\newtheorem{remark}{Remark}
\newtheorem{example}{Example}

\input{peripherals/macros}

\newcommand\acknowledgement[1]{%
  \begingroup
  \renewcommand\thefootnote{}\footnote{#1}%
  \addtocounter{footnote}{-1}%
  \endgroup
}

\usepackage{fancyhdr}
\usepackage{calc}
\fancyheadoffset[RE]{\marginparsep+\marginparwidth}

\usepackage[top=1.5in, bottom=1in, left=1in, right=1in]{geometry}

\graphicspath{ {peripherals/fig/} }

\newif \ifarXiv
\arXivtrue

\ifarXiv
\newcommand{\smi}{\hyperref[app:netflow:tv]{SM-I}}
\newcommand{\smii}{\hyperref[app:netflow:ss]{SM-II}}
\newcommand{\smiii}{\hyperref[sec:crypto]{SM-III}}
\newcommand{\smiv}{\hyperref[app:examples]{SM-IV}}
\newcommand{\smv}{\hyperref[app:assump:honest:justification]{SM-V}}
\newcommand{\smvi}{\hyperref[app:RRA]{SM-VI}}
\newcommand{\smvii}{\hyperref[app:diffpriv_efficacy]{SM-VII}}
\newcommand{\smviii}{\hyperref[app:cong_price]{SM-VIII}}
\newcommand{\smix}{\hyperref[app:merkle_proof]{SM-IX}}
\else
\newcommand{\smi}{SM-I}
\newcommand{\smii}{SM-II}
\newcommand{\smiii}{SM-III}
\newcommand{\smiv}{SM-IV}
\newcommand{\smv}{SM-V}
\newcommand{\smvi}{SM-VI}
\newcommand{\smvii}{SM-VII}
\newcommand{\smviii}{SM-VIII}
\newcommand{\smix}{SM-IX}
\fi 
\title{Trust but Verify: Cryptographic Data Privacy \\for Mobility Management}

\author{

\begin{tabular}{cc}
    \begin{tabular}{c}
        Matthew Tsao \\
        Stanford University \\
        \texttt{mwtsao@stanford.edu} 
    \end{tabular}
     & 
    \begin{tabular}{c}
        Kaidi Yang \\
        Stanford University \\
        \texttt{kaidi.yang@stanford.edu} 
    \end{tabular}
    \\
    & \\
    \begin{tabular}{c}
        Stephen Zoepf \\
        Lacuna Technologies \\
        \texttt{stephen.zoepf@lacuna.ai} 
    \end{tabular}
     & 
    \begin{tabular}{c}
        Marco Pavone \\
        Stanford University \\
        \texttt{pavone@stanford.edu} 
    \end{tabular}    
\end{tabular}

}

\begin{document}
\maketitle

\begin{abstract}
    The era of Big Data has brought with it a richer understanding of user behavior through massive data sets, which can help organizations optimize the quality of their services. In the context of transportation research, mobility data can provide Municipal Authorities (MA) with insights on how to operate, regulate, or improve the transportation network. Mobility data, however, may contain sensitive information about end users and trade secrets of Mobility Providers (MP). Due to this data privacy concern, MPs may be reluctant to contribute their datasets to MA. Using ideas from cryptography, we propose an interactive protocol between a MA and a MP in which MA obtains insights from mobility data without MP having to reveal its trade secrets or sensitive data of its users. This is accomplished in two steps: a commitment step, and a computation step. In the first step, Merkle commitments and aggregated traffic measurements are used to generate a cryptographic commitment. In the second step, MP extracts insights from the data and sends them to MA. Using the commitment and zero-knowledge proofs, MA can certify that the information received from MP is accurate, without needing to directly inspect the mobility data. We also present a differentially private version of the protocol that is suitable for the large query regime. The protocol is verifiable for both MA and MP in the sense that dishonesty from one party can be detected by the other. The protocol can be readily extended to the more general setting with multiple MPs via secure multi-party computation.
\end{abstract}

% Acknowledgments
\acknowledgement{This  research  was  supported  by  the  National  Science  Foundation under CAREER Award CMMI-1454737. K. Yang would like to acknowledge the support of the Swiss  National  Science  Foundation (SNSF) Postdoc Mobility Fellowship (P400P2\_199332).} 

\newpage
% table of contents
\tableofcontents
\newpage 

%%%%%%%%%%%%%%%%%%%%%%%%%%%%%%%%%%%%%%%%%%%%%%%%%%%%%%%%%%%%%%%%%%%%%
% Introduction
%%%%%%%%%%%%%%%%%%%%%%%%%%%%%%%%%%%%%%%%%%%%%%%%%%%%%%%%%%%%%%%%%%%%%

\section{Introduction}

The rise of mobility as a service, smart vehicles and smart cities is revolutionizing transportation industries all over the world. Mobility management, which entails operation, regulation, and innovation of transportation systems, can leverage mobility data to improve the efficiency, safety, accessibility, and adaptability of transportation systems far beyond what was previously achievable. The analysis and sharing of mobility data, however, introduces two key concerns. The first concern is data privacy; sharing mobility data can introduce privacy risks to end users that comprise the datasets. The second concern is credibility; in situations where data is not shared, how can the correctness of numerical studies be verified? These concerns motivate the need for data analysis tools for transportation systems which are both \textit{privacy preserving} and \textit{verifiable}. 

The \textit{data privacy} issue in transportation is a consequence of the trade-off between data availability and data privacy. While user data can be used to inform infrastructure improvement, equity and green initiatives, the data may contain sensitive user information and trade secrets of mobility providers. As a result, end users and mobility providers may be reluctant to share their data with city authorities. Cities have recently begun mandating micromobility providers to share detailed trajectory data of all trips, arguing that the data is needed to enforce equity or environmental objectives. Some mobility providers argued that while names and other directly identifiable information may not be included in the data, trajectory data can still reveal schedules, routines and habits of the city's inhabitants. The mobility providers' concern over the release of anonymized data is justified. \cite{DworkMNS06} showed that any attempt to release anonymized data either fails to provide anonymity, or there are low-sensitivity attributes of the original dataset that cannot be determined from the published version. In general, anonymization is increasingly easily defeated by the very techniques that are being developed for many legitimate applications of big data \cite{PCAST14}. Such disputes highlight the need for privacy-preserving data analysis tools in transportation.

A communication scheme between a sender and a receiver is \textit{verifiable} if it enables the receiver to determine whether the message or report it receives is an accurate representation of the truth. When the objectives of mobility providers and policy makers are not aligned, one party may benefit from misreporting data or other information, giving rise to verifiability issues in transportation. An example of this is Greyball software \cite{Isaac17}. Mobility providers developed Greyball software to deny service or display misleading information to targeted users. It was originally developed to protect their drivers from oppressive authorities in foreign countries, by misreporting driver location to accounts that were believed to belong to the oppressive authorities. However, mobility providers also used Greyball to hide their activity from authorities in the United States when their operations were scrutinized. Another example of verifiability issues is third party wage calculation apps \cite{Szymkowski21}. Drivers, frustrated by instances of being underpaid, created an app to confirm whether the pay was consistent with the length and duration of each trip. Such incidents highlight the need for verifiable data analysis tools in transportation.

%%%%%%%%%%%%%%%%%%%%%%%%%%%%%%%%%%%%%%%%%%%%%%%%%%%%%%%%%%%%%%%%%%%%%
% Introduction: Statement of Contributions 
%%%%%%%%%%%%%%%%%%%%%%%%%%%%%%%%%%%%%%%%%%%%%%%%%%%%%%%%%%%%%%%%%%%%%

\subsection{Statement of Contributions}

In this paper we propose a protocol between a Municipal Authority and a Mobility Provider that enables the Mobility Provider to send insights from its data to the Municipal Authority in a privacy-preserving and verifiable manner. In contrast to \textit{non-interactive} data sharing mechanisms (which are currently used by most municipalities) where a Municipal Authority is provided an aggregated and anonymized version of the data to analyze, our proposed protocol is an \textit{interactive} mechanism where a Municipal Authority sends queries and Mobility Providers give responses. By sharing responses to queries rather than the entire dataset, interactive mechanisms circumvent the data anonymization challenges faced by non-interactive approaches \cite{DworkMNS06, PCAST14}. 

Our proposed protocol, depicted in Figure~\ref{fig:protocol_overview}, has three main steps. In the first step, the Mobility Provider uses its data to produce a data identifier which it sends to the Municipal Authority. The Municipal Authority can then send its data query to the Mobility Provider in the second step. In the third step, the Mobility Provider sends its response along with a zero knowledge proof. The Municipal Authority can use the zero knowledge proof to check that the response is consistent with the identifier, i.e., the response was computed from the same data that was used to create the identifier. If the Municipal Authority has multiple queries, steps 2 and 3 are repeated. %The protocol is computationally efficient due to zk-SNARK implementations \cite{GabizonWC19} for zero knowledge proofs.

The protocol uses cryptographic commitments and aggregated traffic measurements to ensure that the identifier is properly computed from the true mobility data. In particular, any deviation from the protocol by one party can be detected by the other, making the protocol strategyproof for both parties. Given that the identifier is properly computed, the zero knowledge proof then enables the Municipal Authority to verify the correctness of the response without needing to directly inspect the mobility data. Since the Municipal Authority never needs to inspect the mobility data, the protocol is privacy-preserving. 

The protocol can be extended to the more general case of multiple Mobility Providers, each with a piece of the total mobility data. This is done by including a secure multi-party computation in step 3 of the protocol. Answering a large number of queries with our protocol can lead to privacy issues since it was shown in \cite{DinurN03} that a dataset can be reconstructed from many accurate statistical measurements. To address this concern, we generalize the protocol to enable differentially private responses from the Mobility Provider in large query regimes.

\begin{figure}
    \centering
    \includegraphics[width = 0.5\textwidth]{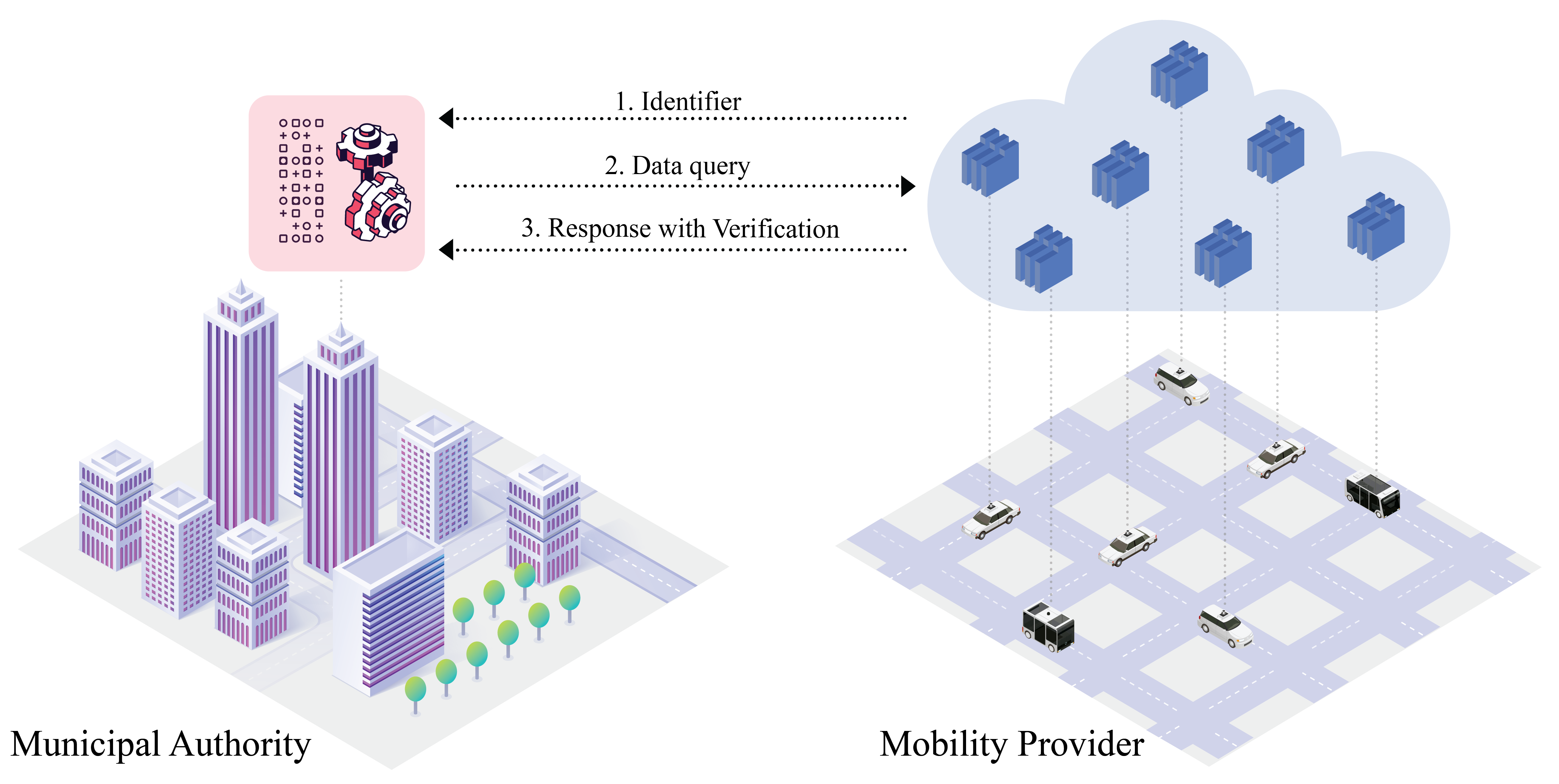}
    \caption{The Mobility Provider can answer the Municipal Authority's data-related mobility queries in a verifiable way \textit{without needing to share the data}. The absence of data sharing in the protocol reduces the chance that a malicious third party intercepts and uses the data for nefarious privacy-invasive purposes.}
    \label{fig:protocol_overview}
\end{figure}

%%%%%%%%%%%%%%%%%%%%%%%%%%%%%%%%%%%%%%%%%%%%%%%%%%%%%%%%%%%%%%%%%%%%%
% Introduction: Organization
%%%%%%%%%%%%%%%%%%%%%%%%%%%%%%%%%%%%%%%%%%%%%%%%%%%%%%%%%%%%%%%%%%%%%

\subsection{Organization}

This paper is organized as follows. The remainder of the introduction discusses academic work related to privacy and verifiability in transportation networks. In Section~\ref{sec:model} we introduce a mathematical model of transportation networks and use it to formulate the data privacy problem for Mobility Management. We provide a high level intuitive description of our proposed protocol in Section~\ref{sec:highlevel}. In Section~\ref{sec:protocol} we provide a full technical description of our protocol. We discuss some of the technical nuances of the protocol and their implications in Section~\ref{sec:discussion}. We summarize our work and identify important areas for future research in Section~\ref{sec:conclusion}. In Appendix~\ref{sec:ext:differentialprivacy} we present a differentially private extension of the protocol that is suitable for the large query regime. 

%%%%%%%%%%%%%%%%%%%%%%%%%%%%%%%%%%%%%%%%%%%%%%%%%%%%%%%%%%%%%%%%%%%%%
% Introduction: Related Work
%%%%%%%%%%%%%%%%%%%%%%%%%%%%%%%%%%%%%%%%%%%%%%%%%%%%%%%%%%%%%%%%%%%%%

\subsection{Related Work}

Within the academic literature, this work is related to the following four fields: misbehavior detection in cooperative intelligent transportation networks, data privacy in transportation systems, differential privacy, and secure multi-party computation. We briefly discuss how this work complements ideas from these fields. 

Cooperative intelligent transportation networks (cITS) aim to provide benefits to the safety, efficiency, and adaptability of transportation networks by having individual vehicles share their information. As with all decentralized systems, security and robustness against malicious agents is essential for practical deployment. As such, misbehavior detection in cITS have been studied extensively \cite{HeijdenDLK19}. Misbehavior detection techniques often rely on honest agents acting as referees, and are able to detect misbehavior in the honest majority setting. Watchdog is one such protocol \cite{Marti00, Hortenlano10} which uses peer-to-peer refereeing. The protocol uses a public key infrastructure (PKI) to assign a persisting identity to each node in the network, and derives a reputation for each node based on its historical behavior. Our objective in this work is also detection of misbehavior, but in a different setting. In our setting, while the mobility network is comprised of many agents (customers and drivers), there is a single entity (the Mobility Provider, e.g., a ridehailing service) who is responsible for the storage and analysis of trip data. As such, the concept of honest majority does not apply to our setting. Furthermore, \cite{Hortenlano10} does not address the issue of data privacy; indeed, PKIs can often expose the users' identities, especially if an attacker cross-references the network traffic with other traffic records.

Privacy in intelligent transportation systems is often implemented by using non-interactive anonymization (e.g., data aggregation), cryptographic tools or differential privacy. Providing anonymity in non-interactive data analysis mechanisms is challenging \cite{DworkMNS06, PCAST14} and thus data aggregation alone is often not enough to provide privacy. From the cryptography side, to address the lack of anonymity provided by blockchains like Bitcoin and Ethereum, zero knowledge proofs \cite{GoldwasserMR89} were deployed in blockchains like Zcash \cite{SassonEtAl14} to provide fully confidential transactions. In the context of transportation, zero knowledge proofs have been proposed for privacy-preserving vehicle authentication to EV charging services \cite{GabayAkkayaEtAl2020}, and privacy-preserving driver authentication to customers in ridehailing applications \cite{LiMeeseEtAl2020}. These privacy-preserving authentication systems rely on a trusted third party to distribute and manage certificates. 

Differential privacy is an interactive mechanism for data privacy which uses randomized responses to hide user-specific information \cite{DworkMNS06}. For any query, the data collector provides a randomized response, where two datasets which differ in only one entry produce statistically indistinguishable outputs. Due to this randomization, there is a trade-off between the accuracy of the response and the level of privacy provided. Randomization is necessary to preserve privacy in the large query regime as demonstrated by \cite{DinurN03} which showed that a dataset can be reconstructed from many accurate statistical measurements. The standard model of differential privacy, however, relies on a \textit{trusted} data collector to apply the appropriate randomized response to queries. This is problematic in situations where the data collector is not trusted. A local model of differential privacy where users perturb their data before sending it to the data collector has received significant attention due to trust concerns \cite{KasiviswanathanLNRS11}. However mobility providers often record exact details about user trips, making local differential privacy unsuitable for current mobility applications (See Remark~\ref{rem:LDP}). Instead, we believe cryptographic techniques can be used to address trust concerns. There are also more general concerns about trust; downstream applications of data queries can lead to conflicts of interest and encourage strategic behavior.

Secure Multi-Party Computation (MPC) is a technique whereby several players, each possessing private data, can jointly compute a function on their collective data without any player having to reveal their data to other players \cite{GoldreichMW87}. MPC achieves confidentiality by applying Shamir's Secret Sharing \cite{Shamir79} to inputs and intermediate results. In its base form, MPC is secure against honest-but-curious adversaries, which follow the protocol, but may try to do additional calculations to learn the private data of other players. In general, security against active malicious adversaries, which deviate from the protocol arbitrarily, requires a trusted third party to perform verified secret sharing \cite{ChorGMA85}. In verified secret sharing, the trusted third party creates initial cryptographic commitments for each player's private data. The commitments do not leak any information about the data, and allows honest players to detect misbehavior using zero knowledge proofs. MPC is a very promising tool for our problem, but a trusted third party able to eliminate strategic behavior does not yet exist in the transportation industry, therefore a key objective of this work is to develop mechanisms to defend against strategic behavior. 

\textit{In Summary -} Our goal in this work is to develop a protocol that enables a mobility provider to share insights from its data to a municipal authority in a privacy-preserving and verifiable manner. Existing work in accountability and misbehavior detection focus on networks with many agents and rely on honest majority. Such assumptions, however, are not realistic for interactions between a municipal authority and a few mobility providers. We thus turn our attention to differential privacy and secure multi-party computation which provide data privacy but require honesty of participating parties. To address this, we develop mechanisms based on cryptography and aggregated roadside measurements to detect dishonest behavior. 

%%%%%%%%%%%%%%%%%%%%%%%%%%%%%%%%%%%%%%%%%%%%%%%%%%%%%%%%%%%%%%%%%%%%%
% Model and Problem Description 
%%%%%%%%%%%%%%%%%%%%%%%%%%%%%%%%%%%%%%%%%%%%%%%%%%%%%%%%%%%%%%%%%%%%%

\section{Model \& Problem Description}\label{sec:model}

In this section we present a model for a city's transportation network and formulate a data Privacy for Mobility Management (PMM) problem. Section~\ref{sec:model:roadnet} introduces a mathematical representation of a city's transportation network along with the demand and mobility providers. In Section~\ref{sec:model:problem} we formalize the notion of data privacy using secure multi-party computation, and introduce assumptions on user behavior that we will need to construct verifiable protocols. We then formally introduce the PMM problem and describe several transportation problems that can be formulated in the PMM framework. 

\subsection{Transportation Network Model}\label{sec:model:roadnet}

\textit{Transportation Network -} Consider the transportation network of a city, which we represent as a directed graph $G = (V,E, f)$ where vertices represent street intersections and edges represent roads. For each road $e \in E$ we use an increasing differentiable convex function $f_e : \mathbb{R}_+ \rightarrow \mathbb{R}_+$ to denote the travel cost (which may depend on travel time, distance, and emissions). of the road as a function of the number of vehicles on the road. We will use $n := \abs{V}$ and $m := \abs{E}$ to denote the total number of vertices and edges in $G$ respectively. Time is represented in discrete timesteps of size $\Delta t$. The operation horizon is comprised of $T+1$ timesteps as $\cT := \bigbrace{0,\Delta t, 2 \Delta t, ..., T \Delta t}$. \\

\noindent \textit{Mobility Provider -} A Mobility Provider (MP) is responsible for serving the transportation demand. It does so by choosing a routing $x$ of its vehicles within the transportation network. The routing must satisfy multi-commodity network flow constraints (see Supplementary Material \smi{} and \smii{} \ifarXiv \else of the extended version \cite{TsaoYangZoepfPavoneE2021} \fi for explicit descriptions of these constraints) and the MP will choose a feasible flow that maximizes its utility function $J_{\text{MP}}$. Some examples of MPs are ridehailing companies, bus companies, train companies, and micromobility (i.e., bikes \& scooters) companies. \\

\noindent \textit{Transportation Demand Data -} The MP's demand data is a list of completed trips $\Lambda := \bigbrace{\lambda_1,...,\lambda_q}$, where $\lambda_i$ contains the following basic metadata about the $i$th trip:
\begin{itemize}
    \item[] Pickup location, Dropoff location, Request time, Match time (i.e., the time at which the user is matched to a driver), Pickup time, Dropoff time, Driver wage, Trip fare, Trip trajectory (i.e., the vehicle's trajectory from the time the vehicle is matched to the rider until the time the rider is dropped off at their destination), Properties of the service vehicle.
\end{itemize}
For locations $i,j \in V$ and a timestep $t$, we use $\Lambda(i,j,t)$ to denote the number of users in the data set who request transit from location $i$ to location $j$ at time $t$. 

\begin{remark}[Multiple Mobility Providers]
We can consider settings where there are multiple mobility providers, $\text{MP}_1, \text{MP}_2,...,\text{MP}_\ell$, where $\Lambda_j$ is the demand data of $\text{MP}_j$. The demand data set for the whole city is thus $\Lambda = \cup_{j=1}^\ell \Lambda_j$. 
\end{remark}

\noindent \textit{Ridehailing Periods -} For MPs that operate ridehailing services, a ridehailing vehicle's trajectory is often divided into three different periods (with Period 0 often ignored):
\begin{itemize}
    \item[] \textit{Period 0:} The vehicle is not online with a platform. The driver may be using the vehicle personally. %maintaining the vehicle
    \item[] \textit{Period 1:} The vehicle is vacant and has not yet been assigned to a rider.
    \item[] \textit{Period 2:} The vehicle is vacant, but it has been assigned to a rider, and is en route to pickup.
    \item[] \textit{Period 3:} The vehicle is driving a rider from its pickup location to its dropoff location.
\end{itemize}

\subsection{Objective: Privacy for Mobility Management (PMM)}\label{sec:model:problem}

In the data Privacy for Mobility Management (PMM) problem, a Municipal Authority (MA) wants to compute a function $g(\Lambda)$ on the travel demand, where $g(\Lambda)$ is some property of $\Lambda$ that can inform MA on how to improve public policies. There are two main obstacles to address: privacy and verifiability. 

Privacy issues arise since trip information may contain sensitive customer information as well as trade secrets of Mobility Providers (MP). For this reason MPs may be reluctant to contribute their data for MA's computation of $g(\Lambda)$. This motivates the following notion of privacy:

\begin{definition}[Privacy in Multi-Party Computation]\label{def:MPCprivacy}
Suppose $\text{MP}_1,...\text{MP}_\ell$ serve the demands $\Lambda_1,...,\Lambda_\ell$ respectively, and we denote $\Lambda = \cup_{i=1}^\ell \Lambda_i$. We say a protocol for computing $g(\Lambda)$ between a MA and several MPs is privacy preserving if 
\begin{enumerate}
    \item MA learns nothing about $\Lambda$ beyond the value of $g(\Lambda)$. 
    \item For any pair $i \neq j$, $\text{MP}_i$ learns nothing about $\Lambda_j$ beyond the value of $g(\Lambda)$.
\end{enumerate}
\end{definition}

Verifiability issues arise if there is incentive misalignment between the players. In particular, if the MA or a MP can increase their utility by deviating from the protocol, then the computation of $g(\Lambda)$ may be inaccurate. To address this issue, we need the protocol to be verifiable, as described by Definition~\ref{def:verifiable}. The following assumption is necessary to ensure accurate reporting of demand (See Supplementary Material~\smv{} \ifarXiv \else of the extended version \cite{TsaoYangZoepfPavoneE2021} \fi for more details):

\begin{assumption}[Strategic Behavior]\label{assump:honest}
We assume in this work that drivers and customers of the transportation network will behave honestly (by this we mean they will always follow the protocol), but MA and MPs may act strategically to maximize their own utility functions. 
\end{assumption}

\begin{definition}[Verifiable Protocol]\label{def:verifiable}
A protocol for computing $g(\Lambda)$ is verifiable under Assumption~\ref{assump:honest} if:
\begin{enumerate}
    \item Any deviation from the protocol by the MA can be detected by the MPs provided that all riders and drivers act honestly (i.e., follow the protocol). 
    \item Any deviation from the protocol by an MP can be detected by the MA provided that all riders and drivers act honestly. 
\end{enumerate}
\end{definition}

\noindent Our objective in this paper is to present a PMM protocol, which is defined below.

\begin{definition}[PMM Protocol]
A PMM protocol between a MA and $\text{MP}_1,...\text{MP}_\ell$ can, given any function $g$, compute $g(\Lambda)$ for MA while ensuring privacy and verifiability as described by Definitions~\ref{def:MPCprivacy} and~\ref{def:verifiable} respectively. 
\end{definition}

\begin{remark}[Admissible Queries and Differential Privacy]\label{rem:admissible_queries}
While a PMM protocol hides all information about $\Lambda$ beyond the value of $g(\Lambda)$, $g(\Lambda)$ itself may contain sensitive information about $\Lambda$. The extreme case would be if $g$ is the identity function, i.e., $g(\Lambda) = \Lambda$. In such a case, the MPs should reject the request to protect the privacy of its customers. More generally, MPs should reject functions $g$ if $g(\Lambda)$ is highly correlated with sensitive information in $\Lambda$. The precise details as to which functions $g$ are deemed acceptable queries must be decided upon beforehand by MA and the MPs together.

Differential privacy mechanisms provide a principled way to address the sensitivity of $g$ by having MPs include noise in the computation of $g(\Lambda)$. If the noise distribution is chosen according to both the desired privacy level and the sensitivity of $g$ to its inputs, then the output is differentially private. Note that this privacy is not for free; the noise reduces the accuracy of the output. The precise choice of noise distribution is important for both the privacy and accuracy of this method, so ensuring that the randomization step is conducted properly in the face of strategic MAs and MPs is essential. This can be done with a combination of coinflipping protocols and secure multi-party computation, which we describe in Appendix~\ref{sec:ext:differentialprivacy}. 
\end{remark}

\begin{remark}[A note on computational complexity]
The applications we consider in this work do not impose strict requirements on computation times of protocols. Regulation checks can be conducted daily or weekly, and infrastructure improvement initiatives are seldom more frequent than one per week. The low frequency of such queries gives plenty of time to compute a solution. For this reason, we do not expect the computational complexity of the solution to be an issue.
\end{remark}

\noindent We now present some important social decision making problems that can be formulated within the PMM framework.

\subsubsection{Regulation Compliance for Mobility Providers}\label{sec:model:problem:regulation}

Suppose MA wants to check whether a MP is operating within a set of regulations $\rho_1,...,\rho_k$. The metadata contained within each trip includes request time, match time, pickup time, dropoff time, and trip trajectory, which can be used to check regulation compliance. If we define the function $\rho_i(\Lambda)$ to be $1$ if and only if regulation $i$ is satisfied, and $0$ otherwise, then regulation compliance can be determined from the function $g(\Lambda) := \prod_{t=1}^k \rho_t(\Lambda)$. Below are some examples of regulations that can be enforced using trip metadata. 

\begin{example}[Waiting Time Equity]
MP is not discriminating against certain requests due to the pickup or droppoff locations. Specifically, the difference in average waiting time among different regions should not exceed a specified regulatory threshold. 
\end{example}

\begin{example}[Congestion Contribution Limit]
The contribution of MP vehicles (in Period 2 or 3) to congestion should not exceed a specified regulatory threshold. 
\end{example}

\begin{example}[Accurate Reporting of Period 2 Miles]\label{regulation:period}
A ridehailing driver's pay per mile/minute depends on which period they are in. In particular, the earning rate for period 2 is often greater than that of period 1. For this reason, mobility providers are incentivized to report period 2 activity as period 1 activity. To protect ridehailing drivers, accurate reporting of period 2 activity should be enforced.
\end{example}

\begin{example}[Emissions Limit]
The collective emission rate of MP vehicles in Phases 2 and 3. should not exceed a specified regulatory threshold. MP emissions can be computed from the metadata of served trips, in particular the trajectory and vehicle make and model. 
\end{example}

\noindent See Supplementary Material~\smiv{} \ifarXiv \else of the extended version \cite{TsaoYangZoepfPavoneE2021} \fi for further details on formulating the above examples within the PMM framework. 

\subsubsection{Transportation Infrastructure Development Projects}\label{sec:model:problem:infrastructure}

\noindent \textit{Transportation Infrastructure Improvment Projects -} A Municipal Authority (MA) measures the efficiency of the current transportation network via a concave social welfare function $J_{\text{MA}}(x)$. The MA wants to make improvements to the network $G$ through infrastructure improvement projects. Below are some examples of such projects.

\begin{example}[Building new roads]
The MA builds new roads $E_\text{new}$ so the set of roads is now $E \cup E_\text{new}$, i.e., $G$ now has more edges. %The impact of additional roads on system utility depends on the demand $\Lambda$. 
\end{example}

\begin{example}[Building Train tracks]
The MA builds new train routes. Train routes differ from roads in that the travel time is independent of the number of passengers, i.e., there is no congestion effect.
\end{example}

\begin{example}[Adding lanes to existing roads]
The MA adds more lanes to some roads $E' \subset E$. As a consequence, the shape of $f_e$ will change for each $e \in E'$. %The impact of additional lanes on utility depends on the demand $\Lambda$. 
\end{example}

\begin{example}[Adjusting Speed limits]
Similar to adding more lanes, adjusting the speed limit of a road will change its delay function. 
\end{example}

\noindent \textit{Evaluation of Projects -} We measure the utility of a project using a Social Optimization Problem (SOP). An infrastructure improvement project $\theta$ makes changes to the transit network, so let $G_\theta$ denote the transit network obtained by implementing $\theta$. The routing problem $\text{ROUTE}(\theta,\Lambda)$ associated with $\theta$ is the optimal way to serve requests in $G_\theta$ as measured by MP's objective function $J_{\text{MP}}$. Letting $S_{\theta,\Lambda}$ be the set of flows satisfying multi-commodity network flow constraints (See Supplementary Material \smi{} and \smii{}  \ifarXiv \else of the extended version \cite{TsaoYangZoepfPavoneE2021} \fi for time-varying and steady state formulations respectively). for the graph $G_\theta$ and demand $\Lambda$, $\text{ROUTE}(\theta,\Lambda)$ is given by
\begin{align}
    \max & \; J_{\text{MP}}(x) \label{eqn:sop:project}\tag{$\text{ROUTE}(\theta,\Lambda)$} \\
    \text{s.t. } & x \in S_{\theta,\Lambda}. \nonumber 
\end{align}

\begin{definition}[The Infrastructure Development Selection Problem]\label{def:InfraDevProb}
Suppose there are $k$ infrastructure improvement projects $\Theta := \bigbrace{ \theta_1,\theta_2,...,\theta_k }$ available, but the city only has the budget for one project. The city will want to implement the project that yields the most utility, which is determined by the following optimization problem. 
\begin{align}
    \underset{1 \leq i \leq k}{\text{argmax }}  J_{\text{MA}} \bigpar{ \underset{x \in S_{\theta_i,\Lambda}}{\text{argmax }} J_{\text{MP}}(x) }. \label{eqn:sop}\tag{SOP$(\Theta, \Lambda)$}
\end{align}
\end{definition}
In the context of PMM, the function $g$ associated with the infrastructure development selection problem is $g(\Lambda) := \text{SOP}(\Theta,\Lambda)$. 

\subsubsection{Congestion Pricing}

Some ridehailing services allow drivers to choose the route they take when delivering customers. When individual drivers prioritize minimizing their own travel time and disregard the negative externalities they place on other travelers, the resulting user equilibrium can experience significantly more congestion than the social optimum. In these cases, the total travel time of the user equilibrium is larger than that of the social optimum. This gap, known as the price of anarchy, is well studied in the congestion games literature. 

Congestion pricing addresses this issue by using road tolls to incentivize self-interested drivers to choose routes so that the total travel time of all users is minimized. The desired road tolls depend on the demand $\Lambda$, so MA would need help from MPs to compute the prices. Congestion pricing can be formulated in the PMM framework through the query function $g_{\text{cp}}$ described in \eqref{eqn:cong_price}.

When the travel cost is the same as travel time, the prices can be obtained from the Traffic Assignment Problem \cite{Sheffi1985}:
\begin{align}\label{opt:so_tt}
    \min & \; \sum_{e \in E} x_e f_e(x_e) \\
    \text{s.t. } & x = \sum_{o \in V} \sum_{d \in V} x^{od} \nonumber \\
    & x^{od} \succeq 0 \; \forall o \in V, d \in V \nonumber \\
    & \sum_{(u,v) \in E} x^{od}_{(u,v)} - x^{od}_{(v,u)} = \Lambda(o,d) \bigpar{ \mathds{1}_{[u = o]} - \mathds{1}_{[u = d]} } \forall u \in V \nonumber
\end{align}
where $x^{od}_e$ denotes the traffic flow from $o$ to $d$ that uses edge $e$.
%where $S_\Lambda$ is the set of flows that satisfy Multi-Commodity Network flow constraints (See Supplementary Materials \smi{} and \smii{} \ifarXiv \else of the extended version \cite{TsaoYangZoepfPavoneE2021} \fi for further details) for the transit network $G$ and demand $\Lambda$. 
The objective measures the sum of the travel times of all requests in $\Lambda$. The desired prices are then given by:
\begin{align}\label{eqn:cong_price}
    g_{\text{cp}}(\Lambda) := \bigbrace{ x_e^* f_e'(x_e^*) }_{e \in E} \text{ where } x^* \text{ solves } \eqref{opt:so_tt}.
\end{align}
See Supplementary Material~\smviii{} \ifarXiv \else of the extended version \cite{TsaoYangZoepfPavoneE2021} \fi for more details on congestion pricing.

%%%%%%%%%%%%%%%%%%%%%%%%%%%%%%%%%%%%%%%%%%%%%%%%%%%%%%%%%%%%%%%%%%%%%
% Intuition for the protocol
%%%%%%%%%%%%%%%%%%%%%%%%%%%%%%%%%%%%%%%%%%%%%%%%%%%%%%%%%%%%%%%%%%%%%

\section{A high level description of the protocol}\label{sec:highlevel}

We focus our discussion on the case where there is one MP. The protocol we will present can be generalized to the multiple MP setting through secure Multi-party Computation \cite{GoldreichMW87}. The simplest way for MA to obtain $g(\Lambda)$ is via a \textit{non-interactive} protocol where MP sends $\Lambda$ to MA. MA could then compute $g(\Lambda)$ and any other attributes of $\Lambda$ that it wants to know. This simple procedure, however, does not satisfy data privacy, since MA now has full access to the demand $\Lambda$. 

To address this concern, one could use an \textit{interactive} protocol where MA sends a description of the function $g$ to MP, MP then computes $g(\Lambda)$ and sends it to MA. This protocol does not require MP to share the demand $\Lambda$. The problem with this approach is that there is no way for MA to check whether MP computed $g(\Lambda)$ properly, i.e., this approach is not verifiable. This is problematic if there is an incentive for MP to act strategically, e.g., if MP wants to maximize its own revenue, rather than social utility. 

In this paper we present a verifiable interactive protocol, which allows MA to check whether or not the message it receives from MP is in fact $g(\Lambda)$. This will result in a protocol where MA is able to obtain $g(\Lambda)$ without requiring MP to reveal any information about $\Lambda$ beyond the value of $g(\Lambda)$. 

First, we describe a non-confidential way to compute $g(\Lambda)$. We will discuss how to make it confidential in the next paragraph. MP will send a commitment $\sigma = \textsf{MCommit}(\Lambda,r)$ of $\Lambda$ to MA. This commitment will enable MA to certify that the result given to it by MP is computed using the true demand $\Lambda$. The commitment is confidential, meaning it reveals nothing about $\Lambda$, and is binding, meaning that it will be inconsistent with any other demand $\Lambda' \neq \Lambda$. Now suppose MP computes a message $z = g(\Lambda)$. To convince MA that the calculation is correct, MP will construct a witness $w := (\Lambda, r)$. When MA receives the message $z$ and witness $w$, it will compute $C(\sigma, z,w)$, where $C$ is an evaluation algorithm. $C(\sigma, z,w)$ evaluates to \texttt{True} if
\begin{enumerate}
    \item Rider Witness and Aggregated Roadside Audit checks are satisfied. ($\sigma$ was reported honestly)
    \item $\textsf{MCommit}(\Lambda, r) = \sigma$. ($\Lambda$ is the demand that was used to compute $\sigma$).
    \item $g(\Lambda) = z$ ($g$ was evaluated properly.)
\end{enumerate}
If any of these conditions are not met, $C(\sigma,z,w)$ will evaluate to \texttt{False}. Finally, MA will accept the message $z$ only if $C(\sigma, z,w) = \texttt{True}$. 

The approach presented in the previous paragraph is not privacy-preserving because the witness $w$ being sent from MP to MA includes the demand $\Lambda$. Fortunately, we can use zero knowledge proofs to obtain privacy. Given an arithmetic circuit $C$ (which in our case is the evaluation algorithm $C$), it is possible for one entity (the prover) to convince another entity (the verifier) that it knows an input $z,w$ so that $C(\sigma,z,w) = \texttt{True}$ without revealing what $w$ is. This is done by constructing a zero knowledge proof $\pi$ from $(z,w)$ and sending $(z,\pi)$ to the verifier instead of sending $(z,w)$. MA can then check whether $\pi$ is a valid proof for $z$. The proof $\pi$ is zero knowledge in the sense that it is computationally intractable to deduce anything about $w$ from $\pi$, aside from the fact $C(\sigma,z,w) = \texttt{True}$. For our application, the prover will be MP who is trying to convince the verifier, which is MA, that it computed $g(\Lambda)$ correctly. 

This protocol requires MP to send a commitment of the true demand data to MA. This is problematic if MP has incentive to be dishonest, i.e., provide a commitment corresponding to a different dataset. To ensure this does not happen, our protocol uses a Rider Witness incentive to prevent MP from underreporting demand, and Aggregated Roadside Audits to prevent MP from overreporting demand. These two mechanisms establish the verifiability of the protocol, since, as seen in first requirement of $C$, MA will reject the message if either of these mechanisms detect dishonesty. 

\textit{In Summary -} We present a verifiable interactive protocol. First, MP sends a commitment of the demand to MA, which ensures that the report is computed using the true demand. The correctness of this commitment is enforced by Rider Witness and Aggregated Roadside Audits. MA then announces the function $g$ that it wants to evaluate. MP computes a message $z \leftarrow g(\Lambda)$ and constructs a witness $w$ to the correctness of $z$. Since $w$ in general contains sensitive information, it cannot be used directly to convince MA to accept the message $z$. MP computes a zero knowledge proof $\pi$ of the correctness of $z$ from $w$, and sends the message $z$ and proof $\pi$ to MA. MA accepts $z$ if $\pi$ is a valid zero knowledge proof for $z$. 

\textit{Implementation -} To implement our protocol we will use several tools from cryptography. The commitment $\sigma$ is implemented as a Merkle commitment. 
%For computing zero knowledge proofs, we will be using PLONK, which is a zero knowledge proof system in which the construction and verification of proofs is computationally efficient. 
For computing zero knowledge proofs, we will need a zk-SNARK that doesn't require a trusted setup. PLONK \cite{GabizonWC19}, Sonic \cite{MallerBKM19}, and Marlin \cite{ChiesaHMMVW20} using a DARK based polynomial commitment schemes described in \cite{BunzFS20, BlockHRRS21}. Other options include Bulletproofs \cite{BunzBBPWM18} and Spartan \cite{Setty20}.
The cryptographic tools used in the protocol are reviewed in Supplementary Material~\smiii{}\ifarXiv \else of the extended version \cite{TsaoYangZoepfPavoneE2021}\fi.

%%%%%%%%%%%%%%%%%%%%%%%%%%%%%%%%%%%%%%%%%%%%%%%%%%%%%%%%%%%%%%%%%%%%%
% The protocol
%%%%%%%%%%%%%%%%%%%%%%%%%%%%%%%%%%%%%%%%%%%%%%%%%%%%%%%%%%%%%%%%%%%%%

\section{The Protocol}\label{sec:protocol}

In this section we present our protocol for the PMM problem described in Section~\ref{sec:model:problem}. For clarity and simplicity of exposition we will focus on the case where there is one Mobility Provider. The single MP case can be extended to the multiple MP case via secure multi-party computation \cite{GoldreichMW87}. We present the protocol, which is illustrated in Figure~\ref{fig:protocol:block}, in Section~\ref{sec:protocol:description}. In Section~\ref{sec:protocol:stratproof} we discuss mechanisms used to ensure verifiability of the protocol. 

The protocol uses the following cryptographic primitives: hash functions, commitment schemes, Merkle trees, public key encryption and zero knowledge proofs. Hash functions map data of arbitrary size to fixed size messages, often used to provide succinct identifiers for large datasets. Commitment schemes are a form of verifiable data sharing where a receiver can reserve data from a sender, obtain the data at a later point, and verify that the data was not changed between the reservation and reception times. A Merkle tree is a particular commitment scheme we will use. In public key encryption, every member of a communication network is endowed with a public key and a private key. The public key is like a mailbox which tells senders how to reach the member, and the secret key is the key to the mailbox, so messages can be viewed only by their intended recipients. Zero knowledge proofs, as discussed in Section~\ref{sec:highlevel}, enable a prover to convince a verifier that it knows a solution to a mathematical puzzle without directly revealing its solution. For a more detailed description of these concepts, we refer the reader to the Supplementary Material \smiii{} \ifarXiv \else of the extended version \cite{TsaoYangZoepfPavoneE2021}\fi, where we provide a self-contained introduction of the cryptographic tools used in this work. % which are described in Supplementary Material~\smiii{} and are discussed in further detail in \cite{BonehShoup20},\cite{NarayananEtAl16}. %All cryptographic tools and primatives used in the protocol are described in Supplementary Material~\smiii{}.

\begin{figure}
    \centering
    \includegraphics[width = \textwidth]{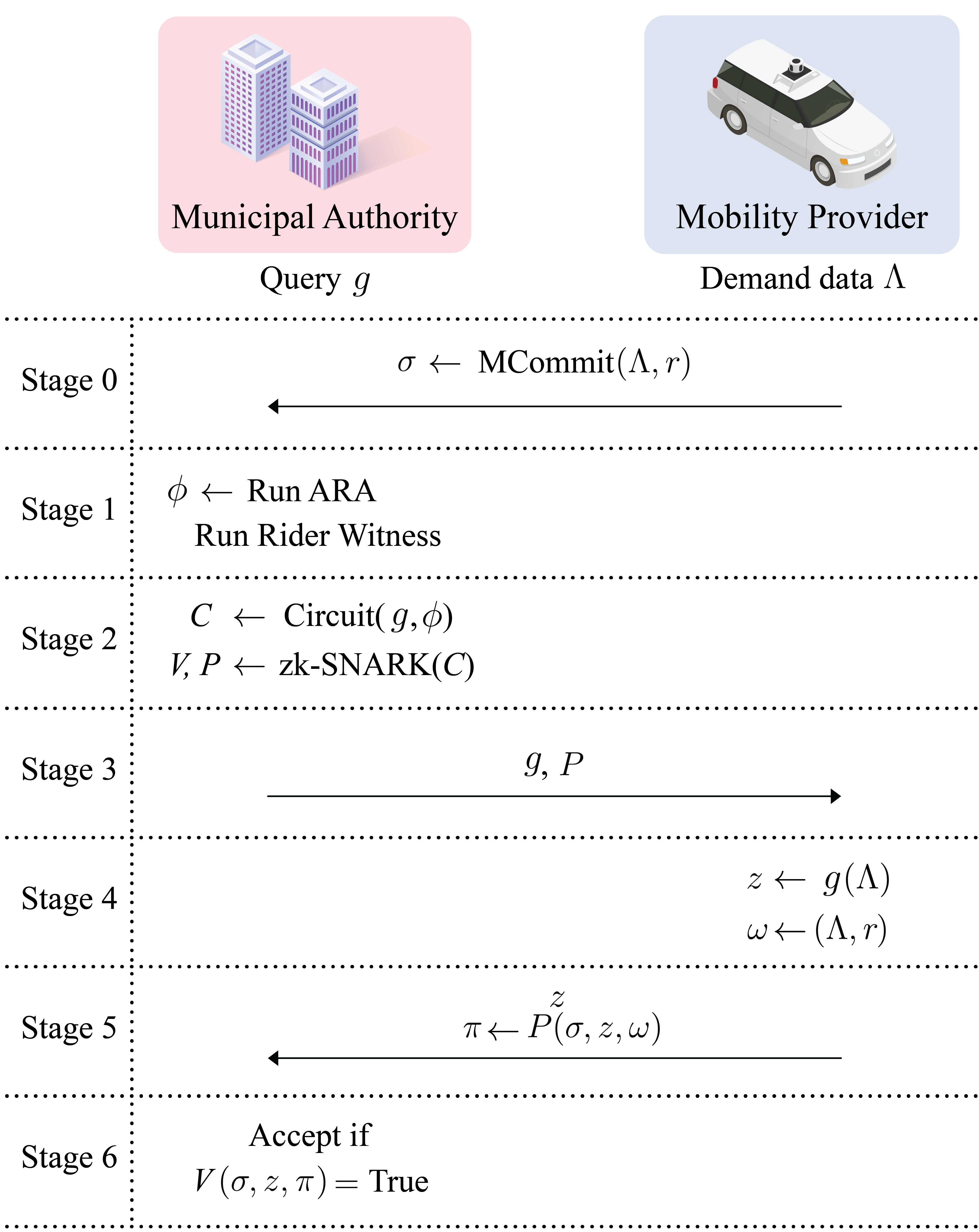}
    \caption{A block diagram of the communication between MA and MP.}
    \label{fig:protocol:block}
\end{figure}

\subsection{Protocol Description}\label{sec:protocol:description}
The protocol entails $6$ stages: \\

\noindent \textit{Stage 0: (Data Collection)} MP serves the demand $\Lambda$ and builds a Merkle Tree $T_\Lambda$ of the demand it serves. MP publishes the root of $T_\Lambda$, which is denoted as $\sigma := \textsf{MCommit}(\Lambda,r)$ so that MA, all riders and all drivers have access to $\sigma$. Here $r$ is the set of nonces used to make the commitment confidential.  \\

\noindent \textit{Stage 1: (Integrity Checks)} MA instantiates Rider Witness and Aggregated Roadside Audits to ensure that $\sigma$ was computed using the true demand $\Lambda$. The description of these mechanisms can be found in Section~\ref{sec:protocol:stratproof}.\\

\noindent \textit{Stage 2: (Message Specifications)} MA specifies to MP the function $g$ it wants to compute. \\

\noindent \textit{Stage 3: (zk-SNARK Construction)} MA constructs an evaluation algorithm $C$ for the function $g$. $\sigma,z$ are public parameters of $C$, and the input to $C$ is a witness of the form $w = (\Lambda_w, r_w, c_w)$, where $r_w$ is a set of nonces, $\Lambda_w$ is a demand matrix, and $c_w$ is an optional input that may depend on $g$ (See Remark~\ref{rem:verify:opt}). $C$ does the following:
\begin{enumerate}
    \item Checks whether the Rider Witness and Aggregated Roadside Audit tests are satisfied (This checks that $\sigma$ was reported honestly),
    \item Checks whether $\textsf{MCommit}(\Lambda_w,r_w) = \sigma$ (This determines whether the provided demand $\Lambda_w$ is the same as the demand that created $\sigma$),
    \item Checks whether $g(\Lambda_w) = z$ (This checks that the message $z$ is computed properly from $\Lambda_w$).
\end{enumerate}
$C$ will evaluate to $\texttt{True}$ if and only if all of those checks pass. Now, using one of the schemes from \cite{GabizonWC19, MallerBKM19, ChiesaHMMVW20, BunzBBPWM18, Setty20}, MA will create a zk-SNARK $(S,V,P)$ for $C$. $S$ is a set of public parameters that describes the circuit $C$, $P$ is a prover function which MP will use to construct a proof, and $V$ is a verification function which MA will use to verify the correctness the MP's proof. It sends $C,(S,V,P), g$ to MP. \\

\noindent \textit{Stage 4: (Function Evaluation)} If the request $g$ is not a privacy-invasive function (see Remark~\ref{rem:admissible_queries}), MP will compute a message $z = g(\Lambda)$ and construct a witness $w := (\Lambda, r, c_w)$ to the correctness of $z$. \\

\noindent \textit{Stage 5: (Creating a Zero Knowledge Proof)} MP uses the zk-SNARK's prover function $P$ to construct a proof $\pi := P(\sigma,z,w)$ that certifies the calculation of $z$. MP sends $z,\pi$ to MA. \\

\noindent \textit{Stage 6: (zk-SNARK Verification)} MA uses the zk-SNARK's verification function $V(\sigma, z,\pi)$ to check whether MP is giving a properly computed message. If this is the case, MA accepts the message $z$. 

\begin{remark}[Computational Gains via Commit-then-Prove]
Steps 2) and 3) of the evaluation circuit $C$ involve different types of computation. This heterogeneity can introduce computational overhead in the zk-SNARK. Commit-and-Prove zk-SNARKs \cite{CampanelliFQ19,CampanelliFFQR20} are designed to handle computational heterogeneities, however existing implementations require a trusted setup. 
%Step 2) checks whether $\Lambda_w, r_w$ is a valid opening of $\sigma$ under a SHA256 based Merkle commitment, and 3) performs an arithmetic computation.
\end{remark}

\begin{remark}[Verifying solutions to convex optimization problems]\label{rem:verify:opt}
If $g(\Lambda_w)$ is the solution to a convex optimization problem parameterized by $\Lambda_w$, (e.g., $g(\Lambda_w) = \text{SOP}(\Theta,\Lambda_w)$ or congestion pricing $g_{\text{cp}}(\Lambda_w)$), then computing $g(\Lambda_w)$ within the evaluation algorithm $C$ may cause $C$ to be a large circuit, thus making evaluation of $C$ computationally expensive. Fortunately, this can be avoided by leveraging the structure of convex problems. If $z = g(\Lambda_w)$, we can include the optimal primal and dual variables associated with $z$ in the optional input $c_w$. This way, checking the optimality of $z$ can be done by checking that $c_w$ satisfy the KKT conditions rather than needing to re-solve the problem. 
\end{remark}

\subsection{Ensuring accuracy of $\sigma$}\label{sec:protocol:stratproof}

The protocol presented in the previous section requires MP to share a commitment to the true demand $\Lambda$. %that is a product of complete and accurate travel behavior. 
However, scenarios exist where the MP may face direct or indirect incentives to misreport demand, such as per-ride fees, congestion charges, or other regulations that may constrain MP operations. In this section we present mechanisms to ensure that MP submits a commitment $\sigma = \textsf{MCommit}(\Lambda,r)$ corresponding to the true demand $\Lambda$ rather than a commitment $\sigma' = \textsf{MCommit}(\Lambda',r)$ corresponding to some other demand $\Lambda'$. Specifically, we present Rider Witness and Aggregated Roadside Audits which detect underreporting and overreporting of demand respectively.%, which is visualized in Figure~\ref{fig:rw_ara_vd}. 

% \begin{figure}[h]
%     \centering
%     \includegraphics[width = 0.35\textwidth]{rw_ara_vd.pdf}
%     \caption{A reported demand $\Lambda'$ differs from $\Lambda$ when either $\Lambda \setminus \Lambda'$ or $\Lambda' \setminus \Lambda$ is non-empty. Under Assumption~\ref{assump:honest}, Rider Witness can detect if $\Lambda \setminus \Lambda'$ (blue) is non-empty, i.e., demand has been underreported. Given that $\Lambda \setminus \Lambda'$ is empty, Aggregated Roadside Audits can detect if $\Lambda' \setminus \Lambda$ (red) is non-empty, i.e., demand has been overreported. Thus together, these mechanisms can detect whether $\sigma$ is a commitment to the true demand $\Lambda$.}
%     \label{fig:rw_ara_vd}
% \end{figure}
\ifarXiv
The Rider Witness mechanism described in Section~\ref{sec:protocol:mech:RW} prevents MP from omitting real trips from its commitment. Under the Rider Witness mechanism, each rider is given a receipt for their trip signed by MP. By signing a receipt, the trip is recognized as genuine by MP. Since $\sigma' = \textsf{MCommit}(\Lambda',r)$ is a Merkle commitment, for each $\lambda' \in \Lambda'$, MP can provide a proof that $\lambda'$ is included in the calculation of $\sigma'$. Conversely, if $\lambda \not\in \Lambda'$, MP is unable to forge a valid proof to claim that $\lambda$ is included in the calculation of $\sigma'$. Therefore if there exists a genuine trip $\lambda \in \Lambda$ that is not included in $\Lambda'$, then that rider can report its receipt to MA. MP cannot provide a proof that $\lambda$ was included, and since the receipt of $\lambda$ is signed by MP, this is evidence that MP omitted a genuine trip from $\sigma'$. If this happens, MP is fined, and the reporting rider is rewarded.

The Aggregated Roadside Audit mechanism described in Section~\ref{sec:protocol:mech:ARA} prevents MP from adding fictitious trips into its commitment. Due to Rider Witness, MP will not omit genuine trips, so $\sigma' = \textsf{MCommit}(\Lambda',r)$ where $\Lambda \subseteq \Lambda'$. Recall that the trip metadata includes the trajectory. If $\Lambda'$ contains fictitious trips, then the road usage reported by $\Lambda'$ will be greater than what happens in reality. Thus if MA measures the number of passenger carrying vehicles that traverse each road, then it will be able to detect if MP has included fictitious trips. However, auditing every road can lead to privacy violations. Therefore, the audits are aggregated so that MA obtains the total volume of passenger carrying traffic in the entire network, but not the per-road traffic information. 
\fi 
\subsubsection{Rider Witness: Detecting underreported demand}\label{sec:protocol:mech:RW}
% Check $\sigma$  for unreported trips

In this section, we present a Rider Witness mechanism to detect omission or tampering of the demand $\Lambda$. Concretely, if a MP sends to MA a Merkle commitment $\sigma' = \textsf{MCommit}(\Lambda',r)$ which underreports demand, i.e., $\Lambda \setminus \Lambda'$ is non-empty, then Rider Witness will enable MA to detect this. MA can impose fines or other penalties when such detection occurs to deter MP from underreporting the demand.

\textit{Rider Witness Incentive Mechanism -} At the beginning of Stage 0 (Data Collection) of the protocol, MP constructs a public key and private key pair $(\textsf{pk}_{\text{mp}}, \textsf{sk}_{\text{mp}})$ to use for digital signatures. The payment process is as follows: When the $i$th customer is delivered to their destination, the customer will send a random nonce $r_i$ to MP. MP will respond with a receipt $\bigpar{ H(r_i || \lambda_i), \sigma_i }$, where $\sigma_i := \textsf{sign}(\textsf{sk}_{\text{mp}}, H(r_i||\lambda_i))$ is a digital signature certifying that MP recognizes $\lambda_i$ as an official ride (here $||$ represents concatenation of binary strings). Here $H$ is SHA256, so that $H(r_i || \lambda_i)$ is a cryptographic commitment to the trip $\lambda_i$. The customer is required to pay the trip fare only if $\textsf{verify}(\textsf{pk}_{\text{mp}}, H(r_i || \lambda_i), \sigma_i) = \texttt{True}$, i.e., they received a valid receipt. 

\begin{definition}[Rider Witness Test]
Given a commitment $\sigma'$ reported by MP to MA, each rider who was served by MP requests a Merkle proof that their ride is included in the computation of $\sigma'$. If there exists a valid\footnote{In the sense that $\textsf{verify}(\textsf{pk}_{\text{mp}}, H(r_i||\lambda_i), \sigma_i) = \texttt{True}$.} ride receipt $(H(r_i||\lambda_i),\sigma_i)$ for which MP cannot provide a Merkle proof, then the customer associated with $\lambda_i$ will report $(H(r_i||\lambda_i),\sigma_i)$ to MA. MA checks if $\sigma_i$ is a valid signature for $H(r_i||\lambda_i)$, and if so, directly asks MP for a Merkle Proof that $\lambda_i$ is included in the computation of $\sigma'$. If MP is unable to provide the proof, then $\sigma'$ fails the Rider Witness Test.
\end{definition}

\begin{observation}[Efficacy of Rider Witness]\label{obv:RW}
Under Assumption~\ref{assump:honest}, if MP submits a commitment $\sigma' = \textsf{MCommit}(\Lambda',r)$ which omits a ride, i.e., $\Lambda \setminus \Lambda'$ is non-empty, then $\sigma'$ will fail the Rider Witness Test. 
\end{observation}

\begin{proof}[Proof of Observation~\ref{obv:RW}]
If $\Lambda \not\subseteq \Lambda'$, then there exists some $\lambda_i$ which is in $\Lambda$ but not $\Lambda'$. Suppose Alice was the rider served by ride $\lambda_i$. Forging a proof that $\lambda_i \in \Lambda'$ requires finding a hash collision for the hash function used in the Merkle commitment. Since $\textsf{MCommit}$ is implemented using a cryptographic hash function (e.g., SHA256), it is computationally intractable to find a hash collision, and thus MP will be unable to forge a valid proof that $\lambda_i \in \Lambda'$. 

If MP does not provide Alice a valid proof within a reasonable amount of time (e.g., several hours), Alice can then report $\bigpar{ H(r_i||\lambda_i),\sigma_i }$ to MA. This reporting does not compromise Alice's privacy due to the hiding property of cryptographic hash functions. MA will check whether $\textsf{verify}(\textsf{pk}_{\text{mp}}, H(r_i || \lambda_i), \sigma_i) = \texttt{True}$, and if so, means that $\lambda_i$ is recognized as a genuine trip by MP. MA will directly ask MP for a Merkle proof that $H(r_i||\lambda_i) \in T_\Lambda$. Since MP cannot provide a valid proof, this is evidence that a genuine trip was omitted in the computation of $\sigma'$, and hence $\sigma'$ will fail the Rider Witness test. 
\end{proof}

%In the event that MP's commitment $\sigma'$ fails the Rider Witness Test, MA can fine MP \$$N$. The purpose of this fine is to deter MP from omitting trips. The \$$N$ can then be given to Alice to encourage customers to report their receipts if they do not receive a Merkle proof from MP. If the fine is chosen large enough, then MP is incentivized to include all trips to avoid paying the fine. 

\begin{remark}[Tamperproof Property]
We note that Rider Witness also prevents the MP from altering the data associated with genuine rides. If MP makes changes to $\lambda_i \in \Lambda$ resulting in some $\lambda_i'$, then by collision resistance of $H$, it is computationally infeasible to find $r'$ so that $H(r_i||\lambda_i) = H(r_i'||\lambda_i')$. If such a change is made, then $H(r_i'||\lambda_i')$ is included into the computation of $\sigma'$ instead of $H(r_i||\lambda_i)$. This means $(H(r_i||\lambda_i),\sigma_i)$ becomes a valid witness that data tampering has occurred.
\end{remark}

\begin{remark}[Receipts are Unforgeable]
Note that it is not possible for a rider to report a fake ride $\lambda' \not\in \Lambda$ to MA. This is because the corresponding signature $\sigma'$ cannot be forged without knowing MP's secret key $\textsf{sk}_{\text{mp}}$. Therefore, assuming $\textsf{sk}_{\text{mp}}$ is only known to MP, only genuine trips can be reported.
\end{remark}

\begin{remark}[Honesty of riders]
The Rider Witness mechanism assumes that riders are honest, i.e., they will not collude with MP by accepting invalid receipts. 
\end{remark}

\subsubsection{Aggregated Roadside Audits: Detecting overreported demand}\label{sec:protocol:mech:ARA}

In this section we present an Aggregated Roadside Audit (ARA) mechanism to detect overreporting of demand. Concretely, if MP announces a commitment $\sigma' = \textsf{MCommit}(\Lambda', r)$, where $\Lambda'$ is a strict superset of $\Lambda$ (i.e., $\Lambda' \setminus \Lambda$ is non-empty), then ARA will enable MA to detect this. Thus between ARA and Rider Witness, MA can detect if MP commits to a demand that is not $\Lambda$. \\

\textit{Aggregated Roadside Audits -} Due to the Rider Witness mechanism, we can assume that MP submits a commitment $\sigma'$ computed from $\Lambda'$ satisfying $\Lambda \subseteq \Lambda'$, i.e., $\Lambda'$ is a superset of $\Lambda$. For an edge $e \in E$ and a demand $\Lambda$, define
\begin{align}\label{eqn:ARA}
    \varphi(e, \Lambda) := \sum_{\lambda \in \Lambda} \mathds{1}_{[\lambda \text{ traverses } e]}
\end{align}
to be the number of trips that traversed $e$ during passenger pickup (Period 2) or passenger delivery (Period 3). Since trip route is provided in the trip metadata, $\varphi(e, \Lambda)$ can be computed from $\Lambda$. 

\begin{definition}[ARA Test]
The Aggregated Roadside Audit places a sensor on every road to conduct an audit on each road $e \in E$ to measure $\varphi(e,\Lambda)$. These values are then aggregated as $\phi := \sum_{e \in E} \varphi(e,\Lambda)$. A witness $w = (\Lambda_w, r_w, c_w)$ passes the ARA test if and only if  
\begin{align}
    \sum_{e \in E} \varphi(e, \Lambda_w) = \phi. \tag{ARA}
\end{align}
\end{definition}

\begin{observation}[Efficacy of Aggregated Roadside Audits]\label{obv:ARA}
Under Assumption~\ref{assump:honest}, if MP submits a commitment $\sigma' = \textsf{MCommit}(\Lambda',r)$ to a strict superset of the demand, i.e., $\Lambda \subset \Lambda'$, then any proof submitted by MP will either be inconsistent with $\sigma'$ or will fail the ARA test. Hence MP cannot overreport demand. 
\end{observation}

\begin{proof}[Proof of Observation~\ref{obv:ARA}]

Suppose $\Lambda'$ is a strict superset of $\Lambda$, which means that there exists some $\lambda' \in \Lambda' \setminus \Lambda$. Then there must exist some $e' \in E$ for which $\varphi(e', \Lambda') > \varphi(e',\Lambda)$. In particular, any edge in the trip route of $\lambda'$ will satisfy this condition. With the inclusion of the ARA test, MP is unable to provide a valid witness for MA's evaluation algorithm $C$ (and as a consequence, will be unable to produce a valid zero knowledge proof) for the following reason: 
\begin{enumerate}
    \item $\textsf{MCommit}$ is a collision-resistant function (since it is built using a cryptographic hash function $H$), so because $\sigma' = \textsf{MCommit}(\Lambda',r)$, it is computationally intractable for MP to find $\Lambda'' \neq \Lambda'$ and nonce values $r''$ so that $\textsf{MCommit}(\Lambda'',r'') = \sigma'$. Therefore, in order to satisfy condition 2 of $C$ (see \textit{Stage 3} of Section~\ref{sec:protocol:description}), MP's witness must choose $\Lambda_w$ to be $\Lambda'$. 
    \item However, $\Lambda'$ will not pass the ARA test. To see this, note that $(a)$ $\Lambda \subseteq \Lambda'$ implies that $\varphi(e,\Lambda) \leq \varphi(e,\Lambda')$ for all $e \in E$. Furthermore, $(b)$ there exists an edge $e'$ where the inequality is strict, i.e., $\varphi(e',\Lambda) < \varphi(e', \Lambda')$. From this, we see that
    \begin{align*}
        \phi = \sum_{e \in E} \varphi(e,\Lambda) &= \varphi(e',\Lambda) + \sum_{e \in \Lambda, e \neq e'} \varphi(e,\Lambda) \\
        &\overset{(a)}{\leq} \varphi(e',\Lambda) + \sum_{e \in \Lambda, e \neq e'} \varphi(e,\Lambda') \\
        &\overset{(b)}{<} \varphi(e',\Lambda') + \sum_{e \in \Lambda, e \neq e'} \varphi(e,\Lambda') \\
        &= \sum_{e \in E} \varphi(e,\Lambda'),
    \end{align*}
    i.e., if the witness passes condition 2 of $C$, then it will fail the ARA test. 
\end{enumerate}
\end{proof}

\begin{figure}
    \centering
    \includegraphics[width = 0.5\textwidth]{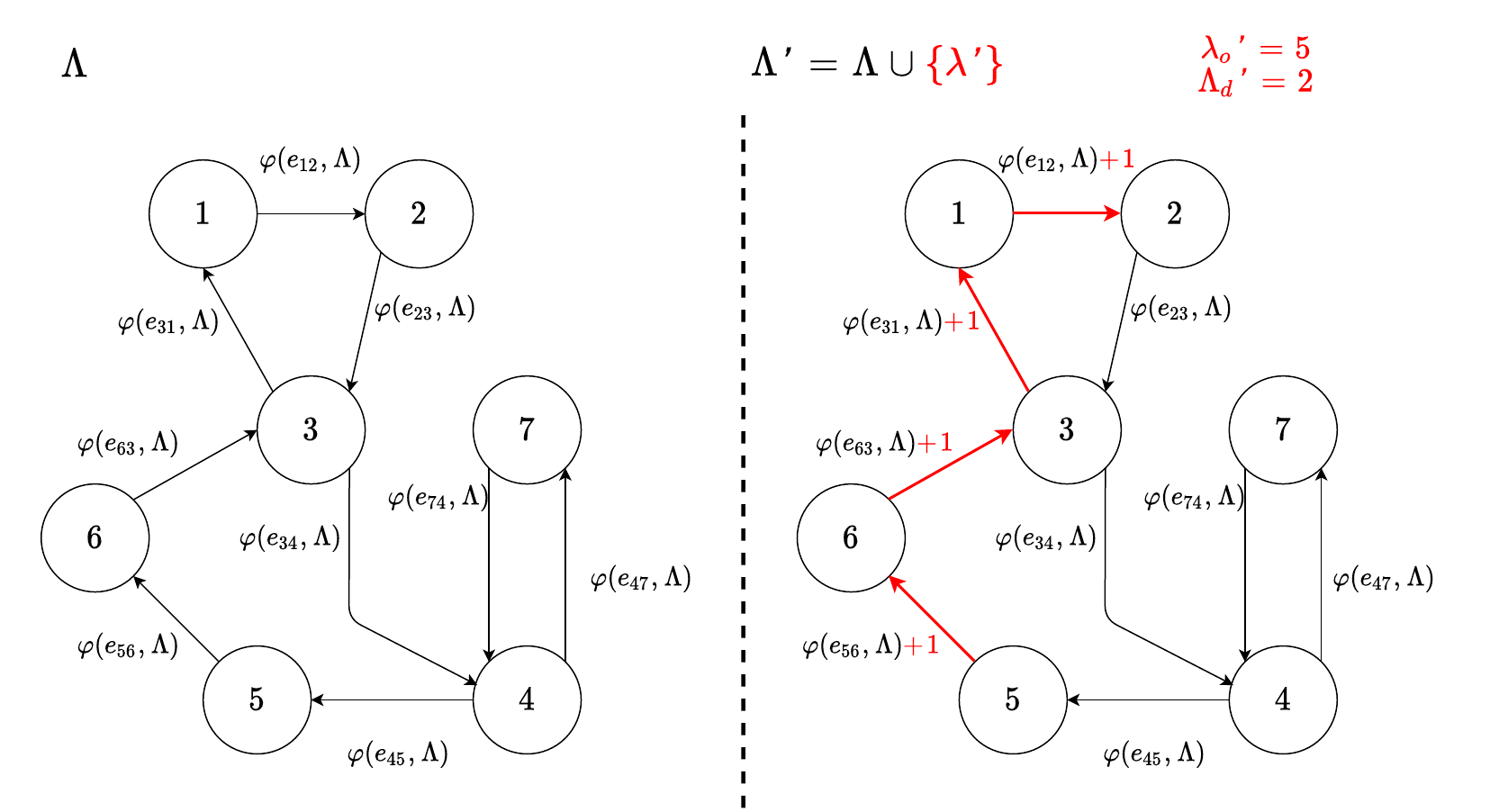}
    \caption{An example of ARA. The true demand is $\Lambda$, which results in traffic shown on the left. Here $\varphi(e_{ij},\Lambda)$ is the total number of trips in $\Lambda$ that use the edge from $i$ to $j$. Suppose MP submits a commitment to $\Lambda' = \Lambda \cup \bigbrace{\lambda'}$, i.e., inserts a fake trip $\lambda'$ into the commitment. In this example, $\lambda'$ is a fake trip from $5$ to $2$ that MP claims was served via the route $\bigbrace{e_{56}, e_{63}, e_{31}, e_{12}}$ (shown in red on the right). $\lambda'$ increases the total traffic on the roads $e_{56}, e_{63}, e_{31}, e_{12}$ and as a result, we have $\sum_{e \in E} \varphi(e,\Lambda') = \phi + 4$.}
    \label{fig:ARA:example}
\end{figure}

Therefore the value of $\phi$ can be used to detect fictitious rides. See Figure~\ref{fig:ARA:example} for a visualization of ARA. In the following remark, we present a variant of ARA that is robust to measurement errors. 

\begin{remark}[Error Tolerance in ARA]
Trip trajectories are often recorded via GPS, so GPS errors can lead to inconsistencies between ARA sensor measurements and reported trajectories. To prevent an honest MP from failing the ARA test due to GPS errors, one can use an error tolerant version of the ARA test defined below
\begin{align*}
    \abs{ \phi - \sum_{e \in E} \varphi(e, \Lambda_w) } \leq \epsilon \phi 
\end{align*}
where $\epsilon \in [0,1]$ is a tuneable tolerance parameter to account for GPS errors while still detecting non-negligible overreporting of demand.
\end{remark}

\begin{remark}[Honesty of Drivers]
The correctness of ARA presented in Observation~\ref{obv:ARA} assumes that drivers are honest when declaring their current period to ARA sensors, e.g., a driver who is in period 3 will not report themselves as period 1 or 2. 
\end{remark}

Two challenges that arise in the computation of $\phi$ are privacy and honesty, which are described below.

\begin{remark}[Privacy-Preserving computation of $\phi$]
%We note that the computation of $\phi$ must be done without compromising data privacy. 
The na{\"i}ve way to compute $\phi$ is for MA to collect the values $\varphi(e,\Lambda)$ from each road. This, however, can compromise data privacy. Indeed, if there is only 1 request in $\Lambda$, then measuring the number of customer carrying vehicles that traverse each link exposes the trip route of that request: Edges that are traversed 1 time are in the route, and edges that are traversed 0 times are not. More generally, observing $\varphi(e,\Lambda)$ on all roads $e \in E$ exposes trip routes to or from very unpopular locations.
\end{remark}

\begin{remark}[Honest computation of $\phi$]
It is essential that MA acts truthfully when taking measurement and computing $\phi$ in ARA, otherwise MP will be wrongfully accused of dishonesty.
\end{remark}

Fortunately, the ARA sensors can use public key encryption to share their data with each other to compute $\phi$ in a privacy-preserving and honest way so that MA cannot learn $\varphi(e,\Lambda)$ for any $e \in E$ even if it tries to eavesdrop on the communication between the sensors. After $\phi$ has been sent to MA and the protocol has finished, the data on the sensors should be erased. We describe this process in Section~\ref{sec:protocol:mech:ARA:sensors}. 

\subsubsection{Implementation details for ARA}\label{sec:protocol:mech:ARA:sensors}

In this section we describe the implementation details of ARA to ensure that the computation of $\phi$ is both privacy-preserving and accurate. \\

\textit{ARA Sensors -} To implement ARA, MA designs a sensor to detect MP vehicles. Concretely, the sensor records the current period of all MP vehicles that pass by. For communication, the sensor will generate a random public and private key pair, and share its public key with the other sensors. The sensor should have hardware to enable it to encrypt and decrypt messages it sends and receives, respectively. To ensure honest auditing by MA, these sensors are inspected by MP to ensure that they detect MP vehicles properly, key generation, encryption and decryption are functioning properly, and that there are no other functionalities. Once the sensors have passed the inspection, the following storage and communication restrictions are placed on them:
\begin{enumerate}
    \item The device can only transmit data if it receives permission from both MP and MA.
    \item The device can only transmit to addresses (i.e., public keys) that are on its sender whitelist. The sender whitelist is managed by both MA and MP, i.e., an address can only be added with the permission of MA and MP. 
    \item The device can only receive data from addresses that are on its receiver whitelist. The receiver whitelist is managed by both MA and MP. 
    \item The device's storage can be remotely erased with permission from both MP and MA.
\end{enumerate}

\textit{Deployment -} To conduct ARA, a sensor is placed on every road, and will record the timestamp and period information of MP vehicles that pass by during the operation period. During operation, both the sender and receiver whitelists should be empty. As a consequence, MA cannot retrieve the sensor data. After the operation period ends and MP has sent a commitment $\sigma$ to MA, MP and MA conduct a coin flipping protocol to choose one sensor at random to elect as a leader (the leader is elected randomly for the sake of robustness. If the leader is the same every time, then the system would be unable to function if this sensor malfunctions or is compromised in any way). A coin flipping protocol is a procedure where several parties can generate unbiased random bits. The leader sensor's public key is added to the whitelist of all other sensors, and all sensors are added to the leader's receiving whitelist. Each sensor then encrypts and sends its data under the leader sensor's public key. Since the MA does not know the leader sensor's secret key, it cannot decrypt the data even if it intercepts the ciphertexts. The addresses of MA and MP are then added to the leader's sender whitelist. The leader sensor decrypts the data, computes $\phi$ and reports the result to both MA and MP. Once the protocol is over, the sender and receiver whitelists of all sensors are cleared, and MA and MP both give permission for the sensors to delete their data. Figure~\ref{fig:ARA} illustrates the sensor setup for ARA. 
\begin{figure}
    \centering
    \includegraphics[width = 0.5\textwidth]{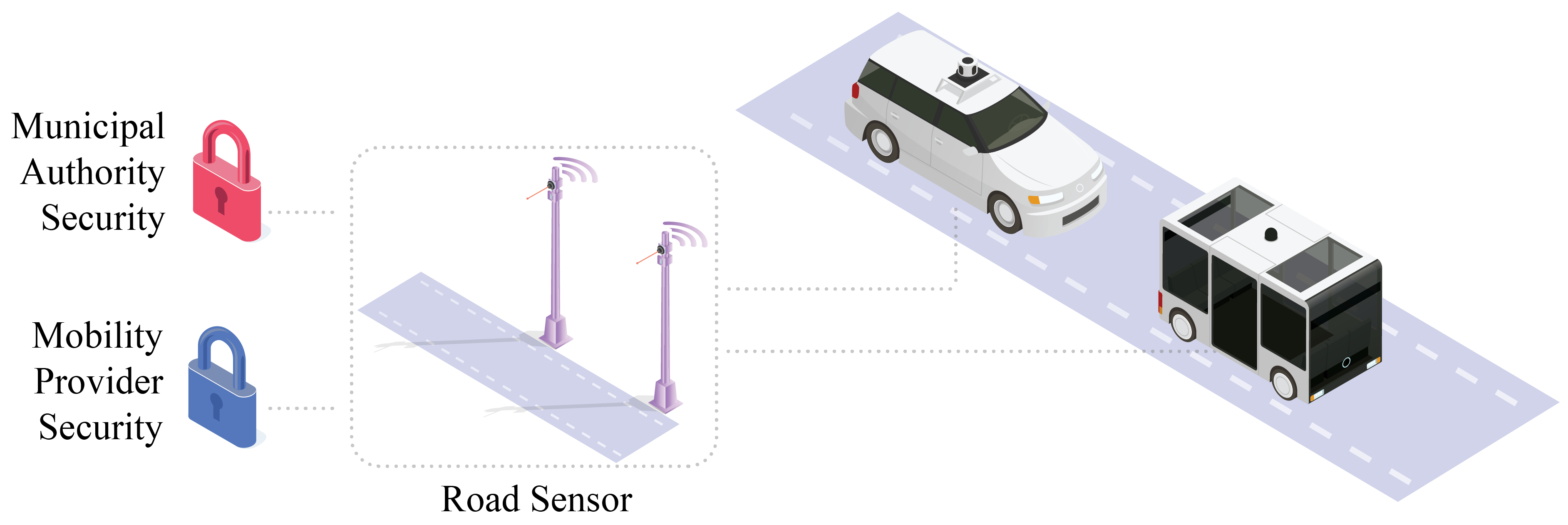}
    \caption{An ARA sensor records the vehicle ID number, vehicle period and current timestamp of each ridehailing vehicle that traverses the road. The dashed line around the sensor represents a communication restriction: The sensor data can only be accessed with the consent of both parties.}
    \label{fig:ARA}
\end{figure}

%%%%%%%%%%%%%%%%%%%%%%%%%%%%%%%%%%%%%%%%%%%%%%%%%%%%%%%%%%%%%%%%%%%%%
% Discussion
%%%%%%%%%%%%%%%%%%%%%%%%%%%%%%%%%%%%%%%%%%%%%%%%%%%%%%%%%%%%%%%%%%%%%

\section{Discussion}\label{sec:discussion}

%In the case that MP is honest, i.e., $\sigma$ is guaranteed to be computed via $\textsf{MCommit}(\Lambda,r)$, PMM can be conducted without revealing any information about $\Lambda$ beyond the value of $g(\Lambda)$. In such a setting, the integrity checks in stage 1 of the protocol can be skipped. However, if MP cannot be assumed to act honestly, then MA must take some measurements to ensure that MP is reporting $\sigma$ truthfully. Through Rider Witness, MA will learn the total number of trips and through ARA, MA learns the total volume of MP traffic during the operation period. Due to the level of aggregation, we do not believe such information would lead to a privacy breach for the mobility users. 

The protocol requires minimal computational resources from the MA. Indeed, the computation of $g(\Lambda)$, and all data analysis therein, is conducted by the MPs. The MA only needs to construct an evaluation circuit $C$ and zk-SNARK $(S,V,P)$ for each of their queries $g$. 
%Both of these steps are computationally efficient. 
In terms of data storage, the MA only needs to store the commitments $\sigma$ to the demand and the total recorded volume of MP traffic $\phi$ for each data collecting period. If the Merkle Trees are built using the SHA256 hash function, then $\sigma$ is only 256 bits, and is thus easy to store. $\phi$ is a single integer, which is also easy to store.  

On the other hand, the hardware requirements for the Aggregated Roadside Audits may be difficult for cities to implement, as placing a sensor on every road in the city will be expensive. To address this concern, we present an alternative mechanism known as Randomized Roadside Audits (RRA) in Supplementary Material~\smvi{} \ifarXiv \else of the extended version \cite{TsaoYangZoepfPavoneE2021} \fi. RRA is able to use fewer sensors by randomly sampling the roads to be audited, however as a tradeoff for using fewer sensors, overreported demand will only be detected probabilistically. See Supplementary Material~\smvi{} \ifarXiv \else of the extended version \cite{TsaoYangZoepfPavoneE2021} \fi for more details. 

There is a trade-off between privacy and diagnosis when using zero knowledge proofs. In the event that the zk-SNARK's verification function fails, i.e., $V(\sigma, z, \pi) = \texttt{False}$, we know that $z$ is not a valid message, but we do not know \textit{why} it is invalid. Specifically, $V(\sigma, z,\pi)$ does not specify which step of the evaluation algorithm $C$ failed (See Stage 3 of Section \ref{sec:protocol:description}). Thus in order to determine whether the failure was due to integrity checks, inconsistency between $\Lambda$ and $\sigma$, or a mistake in the computation of $g$, further investigation would be required. Thus, while the zero knowledge proof enables us to check the correctness of $z$ without directly inspecting the data, it does not provide any diagnosis in the event that $z$ is invalid. 

\ifarXiv
Multi-party computation is a natural way to generalize the proposed protocol to the multiple MP setting. In such a case, the demand $\Lambda = \cup_{i=1}^k \Lambda_k$ is the disjoint union of $\Lambda_1,...,\Lambda_k$, where $\Lambda_i$ is the demand served by the $i$th MP, and is hence the private data of the $i$th MP. Multi-party computation is a procedure by which several players can compute a function over their combined data without any player learning the private data of other players. In the context of PMM with multiple MPs, the MPs are the players and their private data is the $\Lambda_i$'s. In stage 0, each MP would send to MA a commitment to its demand data, and the computation of $z$ and $\pi$ in stages 4 and 5 would be done using secure multi-party computation. Verifiability is established using Rider Witness and ARA, as is done in the single MP case. See \cite{LapetsJAIQVB18} and \url{multiparty.org} for an open-source implementation of multi-party computation. 
\fi 
%%%%%%%%%%%%%%%%%%%%%%%%%%%%%%%%%%%%%%%%%%%%%%%%%%%%%%%%%%%%%%%%%%%%%
% Conclusion
%%%%%%%%%%%%%%%%%%%%%%%%%%%%%%%%%%%%%%%%%%%%%%%%%%%%%%%%%%%%%%%%%%%%%

\section{Conclusion}\label{sec:conclusion}

In this paper we presented an interactive protocol that enables a Municipal Authority to obtain insights from the data of Mobility Providers in a verifiable and privacy-preserving way. During the protocol, a Municipal Authority submits queries and a Mobility Provider computes responses based on its mobility data. The protocol is privacy-preserving in the sense that the Municipal Authority learns nothing about the dataset beyond the answer to its query. The protocol is verifiable in the sense that any deviation from the protocol's instructions by one party can be detected by the other. Verifiability is achieved by using cryptographic commitments and aggregated roadside measurements, and data privacy is achieved using zero knowledge proofs. 
%The protocol is computationally efficient since cryptographic commitments, aggregated roadside measurements, and zero knowledge proofs all have efficient implementations. 
We showed that the protocol can be generalized to a setting with multiple Mobility Providers using secure multi-party computation. We present a differentially private version of the protocol in Appendix~\ref{sec:ext:differentialprivacy} to address situations where the Municipal Authority has many queries. 

There are several interesting and important directions for future work. First, while this work accounts for strategic behavior of the Municipal Authority and Mobility Providers, it assumes that drivers and customers will act honestly. A more general model which also accounts for potential strategic behavior of drivers and customers would be of great value and interest. Second, while secure multi-party computation can be used to generalize the protocol to settings with multiple Mobility Providers, generic tools for secure multi-party computation introduce computational and communication overhead. Developing specialized multi-party computation tools for mobility-related queries is thus of significant practical interest. Finally, we suspect there are other applications for this protocol in transportation research beyond city planning and regulation enforcement that could be investigated. 

\bibliographystyle{ieeetr}
\bibliography{peripherals/anontransit.bib}

\newpage 
\appendix

%%%%%%%%%%%%%%%%%%%%%%%%%%%%%%%%%%%%%%%%%%%%%%%%%%%%%%%%%%%%%%%%%%%%%
% Extensions
%%%%%%%%%%%%%%%%%%%%%%%%%%%%%%%%%%%%%%%%%%%%%%%%%%%%%%%%%%%%%%%%%%%%%

\section{Incorporating Differential Privacy for the Large Query Regime}\label{sec:ext:differentialprivacy}

One potential concern with the protocol described in Section~\ref{sec:protocol} arises in the large query regime. It was shown in \cite{DinurN03} that a dataset can be reconstructed from many accurate statistical measurements. One way to address this is to set a limit on the number of times the MA can query the data for a given time period. Such a restriction would not lead to data scarcity since the MP is collecting new data daily. Differential privacy offers a principled way to determine how many times MA should query a dataset (see Remark~\ref{rem:priv_budget}). Differentially private mechanisms address the result of \cite{DinurN03} by reducing the accuracy of the responses to queries, i.e., responding to a query $g$ with a noisy version of $g(\Lambda)$. In this section we describe how the protocol from section~\ref{sec:protocol} can be generalized to facilitate verifiable and differentially private responses from MP. To this end we first define differential privacy. 

\begin{definition}[Datasets and Adjacency]
A dataset $\Lambda$ is a set of datapoints. In the context of transportation demand, a datapoint is the metadata corresponding to a single trip. We say two datasets $\Lambda,\Lambda'$ are adjacent if either (a) $\Lambda \subset \Lambda'$ with $\Lambda'$ containing exactly 1 more datapoint than $\Lambda$, or (b) $\Lambda' \subset \Lambda$ with $\Lambda$ containing exactly 1 more datapoint than $\Lambda'$.
\end{definition} 

\begin{definition}[Differential Privacy]
Let $\cF$ be a $\sigma$-algebra on a space $\Omega$. A mechanism $M : \cD \rightarrow \Omega$ is $(\epsilon,\delta)$-differentially private if for any two adjacent datasets $\Lambda,\Lambda' \in \cD$ and any $\cF$-measurable event $S$,
\begin{align*}
    \mathbb{P} \bigpar{ M(\Lambda) \in S } \leq e^\epsilon \mathbb{P} \bigpar{ M(\Lambda') \in S } + \delta.
\end{align*}
\end{definition}
In words, the output of a $(\epsilon,\delta)$-differentially private mechanism on $\Lambda$ is statistically indistinguishable from the output of the mechanism on $\Lambda \cup \bigbrace{\lambda}$ for any single datapoint $\lambda \not\in \Lambda$. Since $\Lambda$ does not contain $\lambda$, $M(\Lambda)$ does not reveal any information about $\lambda$. Since $M(\Lambda \cup \bigbrace{\lambda} )$ is statistically indistinguishable from $M(\Lambda)$, $M(\Lambda \cup \bigbrace{\lambda} )$ does not reveal much about $\lambda$. 

\begin{example}[Laplace Mechanism for Vote Tallying]\label{ex:laplace_voting}
Suppose a city is trying to decide whether to expand its railways or expand its roads based on a majority vote from its citizens. The dataset is $\Lambda := \bigbrace{\lambda_1,...,\lambda_n}$ where $\lambda_i$ is a boolean which is $0$ if the $i$th citizen prefers the railway and $1$ if the $i$th citizen prefers the roads. To implement majority vote, the city needs to compute $g(\Lambda) := \sum_{i=1}^n \lambda_i$. The Laplace Mechanism achieves $(\epsilon,0)$-differential privacy for this computation via
\begin{align*}
    M_{\text{laplace}}(\Lambda) := Y + \sum_{i=1}^n \lambda_i,
\end{align*}
where $Y$ has the discrete Laplace distribution: for any $k \in \mathbb{Z}$, $\mathbb{P}[Y = k] \propto e^{-\epsilon \abs{k}}$. To see why this achieves $(\epsilon,0)$-differential privacy, for any $1 \leq j \leq n$, note that
\begin{align*}
    \frac{\mathbb{P}[M(\Lambda) = k]}{\mathbb{P}[M(\Lambda \setminus \bigbrace{\lambda_j}) = k]} = \frac{e^{-\epsilon \abs{k - \sum_{i=1}^n \lambda_i}}}{e^{-\epsilon \abs{ k - \sum_{i \neq j} \lambda_i }}} \leq e^{\epsilon \lambda_j} \leq e^\epsilon. 
\end{align*}
Note that the noise distribution for $Y$ depends only on $\epsilon$, and is independent of $n$, the size of the dataset. 
\end{example}

\begin{remark}[Privacy Budget]\label{rem:priv_budget}
By composition rules, the result of $k$ queries to a $(\epsilon,0)$-differentially private mechanism is $(k\epsilon,0)$-differentially private. Thus a dataset should only be used to answer $k$ separate $(\epsilon,0)$-differentially private queries if $e^{k\epsilon}$ is sufficiently close to $1$.
\end{remark}

\subsection{Goal: Differential Privacy without Trust}
Given a query function $g$ from MA, let $M$ be an polynomial-time computable $(\epsilon,\delta)$-differentially private mechanism for computing $g$. For a given dataset $\Lambda$ we can represent the random variable $M(\Lambda)$ with a function $\widetilde{g}(\Lambda,Z)$ where $Z \in \bigbrace{0,1}^{v}$ represents the random bits used by $M$. Here $v$ is an upper bound on the number of random bits needed for the computation of $M$. By its construction, $\widetilde{g}(\Lambda,Z)$ is $(\epsilon,\delta)$-differentially private if $Z$ is drawn uniformly at random over $\bigbrace{0,1}^v$. Therefore differential privacy is achieved if MP draws $Z$ uniformly at random over $\bigbrace{0,1}^v$ and sends $\widetilde{g}(\Lambda,Z)$ to MA. However, as mentioned in Assumption~\ref{assump:honest}, we are studying a model where MP can act strategically. Thus we cannot assume that MP will sample $Z$ uniformly at random if there is some other distribution over $Z$ that leads to a more favorable outcome for MP. We revisit Example~\ref{ex:laplace_voting} to illustrate this concern. 

% (by efficiently computable we mean: if $\tau(\Lambda)$ is the time needed to compute $g(\Lambda)$, then the time needed to compute $M(\Lambda)$ should be polynomial in $\tau(\Lambda)$. This also implies that the number of random bits $v$ needed by $M$ is also polynomial in $\tau(\Lambda)$).

\begin{example}[Dishonest Vote Tallying]
Consider the setting from Example~\ref{ex:laplace_voting}. The Laplace mechanism can be represented as
\begin{align*}
    \widetilde{g}(\Lambda,Z) := Y + \sum_{i=1}^n \lambda_i, \text{ where } Y = F_{\text{laplace}}^{-1} \bigpar{ \frac{\text{int}(Z)}{2^v} },
\end{align*}
where $\text{int}(Z)$ is the integer whose binary representation is the bits of $Z$. Here $F_{\text{laplace}}^{-1}$ is the inverse cumulative distribution for the discrete Laplace distribution. Thus $F_{\text{laplace}}^{-1}(\text{int}(Z)/2^v)$ is an application of inverse transform sampling that converts a uniform random variable $Z$ into a random variable $Y$ with a discrete Laplace distribution. Suppose the MP has a ridehailing service and would thus prefer an upgrade to city roads over an upgrade to the railway system. If this is the case, choosing $Z$ so that $\widetilde{g}(\Lambda, Z) > n/2$ (as opposed to choosing $Z$ randomly) is a weakly dominant strategy for MP, even if $g(\Lambda) < n/2$ and a majority of the citizens prefer railway upgrades. 
\end{example}

Thus we need a way to verify that the randomness $Z$ used in MP's evaluation of $g(\Lambda,Z)$ has the correct distribution. We will now show how the protocol can be adjusted to accommodate this, and as a consequence, enable verifiable differentially private data queries for MA.

\begin{remark}[MA provided randomness]\label{rem:MA_randomness}
One natural attempt to ensure that $Z$ is uniformly random is to have MA specify $Z$. However, this destroys the differential privacy, since for some mechanisms (including the Laplace mechanism) $g(\Lambda)$ can be computed from $\widetilde{g}(\Lambda,Z)$ and $Z$. Also, it is not clear a priori whether such a setup is strategyproof for MA. 
\end{remark}

\subsection{A Differentially Private version of the protocol}

In this section, we present modifications to the protocol from Section~\ref{sec:protocol:description} that enables verifiable differentially private responses from MP. At a high level, the MA and MP jointly determine the random bits $Z$ via a coin flipping protocol \cite{Blum82}. The zk-SNARK can then be modified to ensure that $\widetilde{g}(\Lambda,Z)$ is computed correctly. The protocol has a total of 6 stages which are described below. \\

\noindent \textit{Stage 0: (Data Collection)} MP builds a Merkle Tree $T_\Lambda$ of the demand $\Lambda$ that it serves. It computes a commitment $\sigma := \textsf{MCommit}(\Lambda,r)$ to this demand. Additionally, MP samples $Z_{\text{mp}}$ uniformly at random from $\bigbrace{0,1}^v$ and computes a Pedersen commitment \cite{Pedersen91} $z_{\text{mp}} := \textsf{Commit}(Z_{\text{mp}},r_{\text{mp}})$. The Pederson commitment scheme is a secure commitment scheme which is perfectly hiding and computationally binding. MP sends both $\sigma, z_{\text{mp}}$ to MA. \\

\noindent \textit{Stage 1: (Integrity Checks)} Same as in Section~\ref{sec:protocol:description}. \\

\noindent \textit{Stage 2: (Message Specifications)} MA specifies the function $g$ it wants to compute. Additionally, MA samples $Z_{\text{ma}}$ uniformly at random from $\bigbrace{0,1}^v$ and specifies a differentially private mechanism $\widetilde{g}$ for the computation of $g$. \\

\noindent \textit{Stage 3: (zk-SNARK Construction)} MA constructs an evaluation circuit $C$ for the function $\widetilde{g}$. The public parameters of $C$ are $\sigma, z_{\text{mp}}, Z_{\text{ma}}, z$ and the input to $C$ is a witness of the form $w = (\Lambda_w, r_w, c_w, Z_{\text{mp},w}, r_{\text{mp},w})$. $C$ does the following:
\begin{enumerate}
    \item Checks whether the Rider Witness and Aggregated Roadside Audit tests are satisfied,
    \item Checks whether $\textsf{MCommit}(\Lambda_w, r_w) = \sigma$,
    \item Checks whether $\textsf{Commit}(Z_{\text{mp},w}, r_{\text{mp},w}) = z_{\text{mp}}$,
    \item Checks whether $\widetilde{g}(\Lambda_w, Z_{\text{ma}} \oplus Z_{\text{mp},w}) = z$. (Here $\oplus$ is bit-wise XOR.) 
\end{enumerate}
$C$ will return $\texttt{True}$ if and only if all of these checks pass. MA constructs a zk-SNARK $(S,V,P)$ for $C$ and sends $g,\widetilde{g},Z_{\text{ma}}, C, (S,V,P)$ to MP. \\

\noindent \textit{Stage 4: (Function Evaluation)} If $\widetilde{g}$ is a differentially private mechanism for computing $g$, then MP computes a message $z = \widetilde{g}(\Lambda, Z_{\text{ma}} \oplus Z_{\text{mp}})$ and a witness $w := (\Lambda,r, c_w, Z_{\text{mp}}, r_{\text{mp}})$ to the correctness of $z$. \\

\noindent \textit{Stage 5: (Creating a Zero Knowledge Proof)} Same as in Section~\ref{sec:protocol:description}. \\

\noindent \textit{Stage 6: (zk-SNARK Verification)} Same as in Section~\ref{sec:protocol:description}. \\

\noindent In Supplementary Material~\smvii{} \ifarXiv \else of the extended version \cite{TsaoYangZoepfPavoneE2021} \fi we show that this protocol has the following two desirable features that enable verifiable and differentially private responses from MP to MA queries. 
\begin{enumerate}
    \item Verifiability - If the MA receives a valid proof from MP, then it can be sure that the corresponding message is indeed $\widetilde{g}(\Lambda, Z_{\text{ma}} \oplus Z_{\text{mp}})$. 
    \item Differential Privacy - The MP's output is differentially private with respect to the dataset $\Lambda$ if at least one of $Z_{\text{ma}},Z_{\text{mp}}$ is sampled uniformly at random.
\end{enumerate}

\begin{remark}[A note on Local Differential Privacy]\label{rem:LDP}
Local Differential Privacy \cite{KasiviswanathanLNRS11} addresses the setting where the data collector is untrusted. Differential privacy is achieved by users adding noise to their data before sending it to the data collector. This is in contrast to the setting we study here where an untrusted data collector has the clean data of many users. We chose to study the latter model due to the way current mobility companies collect high resolution data on the trips they serve. Additionally, local differential privacy requires users to add noise to their data so they become statistically indistinguishable from one another. In the context of transportation, this means the noisy data of users will be statistically indistinguishable from one another, even if they have very different travel preferences. This level of noise significantly reduces the accuracy of any computation done on the data. 
\end{remark}

%%%%%%%%%%%%%%%%%%%%%%%%%%%%%%%%%%%%%%%%%%%%%%%%%%%%%%%%%%%%%%%%%%%%%
% Supplementary Material
%%%%%%%%%%%%%%%%%%%%%%%%%%%%%%%%%%%%%%%%%%%%%%%%%%%%%%%%%%%%%%%%%%%%%
\ifarXiv
\newpage 

\section{Supplementary Material}

\subsection{Mobility Provider Serving Demand}\label{app:netflow:tv}

For a given discretization of time $\cT := \bigbrace{0, \Delta t, 2\Delta t, ..., T \Delta t}$, the demand $\Lambda \in \mathbb{R}^{n \times n \times T}$ can be represented as a 3-dimensional matrix (e.g., a 3-Tensor) where $\Lambda(i,j,t)$ represents the number of riders who request transit from $i$ to $j$ at time $t$. We use $\tau_{ij}$ to represent the time it takes to travel from $i$ to $j$. %We use $\tau_{ijt}$ to denote the arrival time of a trip from $i$ to $j$ that begins at time $t$. 

To serve the demand from $i$ to $j$, the MP chooses passenger carrying flows $x^{ij} \in \mathbb{R}_+^{mT}$ where $x^{ij}_t(u,v)$ is the number of passenger carrying trips from $i$ to $j$ that enter the road $(u,v)$ at time $t$. Such vehicles will exit the road at time $t + \tau_{uv}$. There is also a rebalancing flow $r \in \mathbb{R}^{mT}$ which represents the movement of vacant vehicles that are re-positioning themselves to better align with future demand. Concretely, $r_t(u,v)$ is the number of vacant vehicles which enter road $(u,v)$ at time $t$. The initial condition is $y \in \mathbb{R}_+^n$, where $y_i$ denotes the number of vehicles at location $i$ at time $0$.

The Mobility Provider's routing strategy is thus $x := \bigpar{ \bigbrace{x^{ij}}_{(i,j) \in V \times V}, r}$ which satisfies the following multi-commodity network flow constraints:

\begin{align}
    \sum_{v : (v,u) \in E} \bigpar{ r_{t-\tau_{vu}}(v,u) + \sum_{(i,j) \in V \times V} x^{ij}_{t-\tau_{vu}}(v,u)} = \sum_{v : (u,v) \in E} \bigpar{r_t(u,v) + \sum_{(i,j) \in V \times V} x^{ij}_t(u,v)}& \label{eqn:netflow:tv:conservation:all}\\ \text{ for all } (u,t) \in V \times [T]& \nonumber 
\end{align}
\begin{align}
    \sum_{v : (u,v) \in E} x^{ij}_t(u,v) = \sum_{v : (v,u) \in E} x^{ij}_{t-\tau_{vu}}(v,u) \text{ for all } (i,j) \in V \times V, t \in [T], u \not\in \bigbrace{i,j}\label{eqn:netflow:tv:conservation:route}
\end{align}
\begin{align}
    \sum_{\tau=0}^t \bigpar{ \sum_{v : (i,v) \in E} x^{ij}_\tau(i,v) - \sum_{v : (v,i) \in E} x^{ij}_{\tau-\tau_{vi}}(v,i)} \leq \sum_{\tau=0}^t \Lambda(i,j,\tau) \text{ for all } (i,j,t) \in V \times V \times [T]\label{eqn:netflow:tv:conservation:source}
\end{align}
\begin{align}
    x^{ij}_t(j,v) = 0 \text{ for all } (i,j) \in V \times V, t \in [T], (j,v) \in E. \label{eqn:netflow:tv:conservation:dest}
\end{align}
\begin{align}
    \sum_{j : (i,j) \in E} x_0^{ij} = y_i \text{ for all } i \in V. \label{eqn:netflow:tv:conservation:initial}
\end{align}

Here \eqref{eqn:netflow:tv:conservation:all} represents conservation of vehicles, \eqref{eqn:netflow:tv:conservation:route},\eqref{eqn:netflow:tv:conservation:source},\eqref{eqn:netflow:tv:conservation:dest} enforce pickup and dropoff constraints according to the demand $\Lambda$, and \eqref{eqn:netflow:tv:conservation:initial} enforces initial conditions. 

The utility received by the Mobility provider (e.g., total revenue) from implementing flow $x$ for a given demand $\Lambda$ is $J_{\text{MP}}(x;\Lambda)$. An optimal routing algorithm for demand $\Lambda$ is a solution to the following optimization problem.
\begin{align*}
    \underset{x}{\text{maximize }} & J_{\text{MP}}(x;\Lambda) \\
    \text{s.t. } & \eqref{eqn:netflow:tv:conservation:all},\eqref{eqn:netflow:tv:conservation:route},\eqref{eqn:netflow:tv:conservation:source},\eqref{eqn:netflow:tv:conservation:dest},\eqref{eqn:netflow:tv:conservation:initial}.
\end{align*}

\subsection{Mobility Provider Serving Demand (Steady State)}\label{app:netflow:ss}

In a steady state model, the demand can be represented as $\Lambda \in \mathbb{R}_+^{n \times n}$, a matrix where $\Lambda(i,j)$ represents the rate at which riders request transit from node $i$ to node $j$.

For each origin-destination pair $(i,j) \in V \times V$, the MP serves the demand $\Lambda(i,j)$ by choosing a passenger carrying flow $x^{ij}_p \in \mathbb{R}_+^{m}$ and a rebalancing flow $x_r \in \mathbb{R}_+^m$ so that $x := \bigpar{\bigbrace{x_p^{ij}}_{(i,j) \in V \times V}, x_r}$ satisfy the multi-commodity network flow constraints with demand $\Lambda$: 

\begin{align}
    &\sum_{v: (u,v) \in E} \bigpar{ x_r(u,v) + \sum_{(i,j) \in V \times V} x^{ij}_p (u,v)} = \sum_{v: (v,u) \in E} \bigpar{ x_r(v,u) + \sum_{(v,u) \in V \times V} x^{ij}_p (v,u)} \text{ for all } u \in V \label{eqn:netflow:ss:conservation:all} \\
    &\Lambda(i,j) \mathds{1}_{[u = j]} + \sum_{v: (u,v) \in E} x^{ij}_p (u,v) = \Lambda(i,j) \mathds{1}_{[u = i]} + \sum_{v: (v,u) \in E} x_r(v,u) + x^{ij}_p (v,u) \text{ for all } (i,j) \in V \times V, u \in V\label{eqn:netflow:ss:conservation:pair}.
\end{align}
Here \eqref{eqn:netflow:ss:conservation:all} represents conservation of flow and \eqref{eqn:netflow:ss:conservation:pair} enforces pickup and dropoff constraints according to the demand $\Lambda$. 

The utility received by the Mobility provider (e.g., total revenue) from implementing flow $x$ is $J_{\text{MP}}(x)$. Therefore the Mobility Provider will choose $x$ according to the following program. 
\begin{align*}
    \underset{x}{\text{maximize }} & J_{\text{MP}}(x) \\
    \text{s.t. } & \eqref{eqn:netflow:ss:conservation:all}, \eqref{eqn:netflow:ss:conservation:pair}
\end{align*}

\subsection{Cryptographic Tools}\label{sec:crypto}

In this section we introduce existing cryptographic tools that are used in the protocol. The contents of this section are discussed in greater detail in \cite{NarayananEtAl16,BonehShoup20, GoldwasserMR89,SassonEtAl14,GabizonWC19}. Throughout this paper, we use $r||x$ to denote the concatenation of $r$ and $x$. 
\subsubsection{Cryptographic Hash Functions}

\begin{definition}[Cryptographic Hash Functions]
A function $H$ is a $d$-bit cryptographic hash function if it is a mapping from binary strings of arbitrary length to $\bigbrace{0,1}^d$ and has the following properties:
\begin{enumerate}
    \item It is deterministic.
    \item It is efficient to compute. 
    \item $H$ is \textit{collision resistant} - For sufficiently large $d$, it is computationally intractable to find distinct inputs $x_1,x_2$ so that $H(x_1) = H(x_2)$.
    \item $H$ is \textit{hiding} - If $r$ is a sufficiently long random string (256 bits is often sufficient), then it is computationally intractable to deduce anything about $x$ by observing $H(r||x)$. 
\end{enumerate}
\end{definition}

Property 3 is called collision resistance and enables the hash function to be used as a digital fingerprint. Indeed, since it is unlikely that two files will have the same hash value, $H(x)$ can serve as a unique identifier for $x$. We refer the interested reader to \cite{NarayananEtAl16} for further details on cryptographic hash functions.

\href{https://web.archive.org/web/20160330153520/http://www.staff.science.uu.nl/~werkh108/docs/study/Y5_07_08/infocry/project/Cryp08.pdf}{SHA256} is a widely used collision resistant hash function which has extensive applications including but not limited to establishing secure communication channels, computing checksums for online downloads, and computing proof-of-work in Bitcoin. 

\subsubsection{Cryptographic Commitments}

Cryptographic commitment schemes are tamper-proof communication protocols between two parties: a sender and a receiver. In a commitment scheme, the sender chooses (i.e., commits to) a message. At a later time, the sender reveals the message to the receiver, and the receiver can be sure that the message it received is the same as the original message chosen by the sender.

\textit{Intuition -} We can think of a commitment scheme as follows: A sender places a message into a box and locks the box with a key. The sender then gives the locked box to the receiver. Once the sender has given the box away, the sender can no longer change the message inside the box. At this point, the receiver, who does not have the key, cannot open the box to read the message. At a later time, the sender can give the key to the receiver, allowing the receiver to read the message. 

A commitment scheme is specified by a message space $\cM$, nonce space $\cR$, commitment space $\cX$, a commitment function $\textsf{commit} : \cM \times \cR \rightarrow \cX$ and a verification function $\textsf{verify}: \cM \times \cR \times \cX \rightarrow \bigbrace{0,1}$. Creating a commitment to a message $m \in \cM$ happens in two steps:

\begin{enumerate}
    \item \textit{Commitment Step -} The sender computes $\sigma := \textsf{commit}(m,r)$ for some $r \in \cR$ and gives $\sigma$ to the receiver.
    \item \textit{Reveal Step -} At some later time, the sender gives $m,r$ to the receiver who accepts $m$ as the original message if and only if $\textsf{verify}(m,r,\sigma) := \mathds{1}_{[\textsf{commit}(m,r) = \sigma]}$ evaluates to $1$. 
\end{enumerate}

\noindent A secure commitment scheme has two important properties:

\begin{enumerate}
    \item \textit{Binding -} If $\sigma$ is a commitment to a value $m$, it is computationally intractable to find $m',r'$ so that $m' \neq m$ and $\textsf{commit}(m',r') = \sigma$. Hence $\sigma$ binds the committer to the value $m$. 
    \item \textit{Hiding -} It is computationally intractable to learn anything about $m$ from $\sigma$. 
\end{enumerate}

Cryptographic hash functions can be used to build secure commitment schemes. To do this, given a cryptographic hash function $H$, we define 
\begin{align*}
    \textsf{commit}(m,r) := H(r||m) \text{ and } \textsf{verify}(m,r,\sigma) := \mathds{1}_{[H(r||m) = \sigma]}.
\end{align*}

The security of this commitment scheme comes from the properties of $H$. The binding property of this commitment scheme follows directly from collision resistance of $H$. Furthermore, if $r$ is chosen uniformly at random from $\cR$, then the commitment scheme is hiding due to the hiding property of $H$. 

\subsubsection{Merkle Trees}

A Merkle tree is a data structure that is used to create commitments to a collection of items $M := \bigbrace{m_0,...,m_{q-1}}$. A Merkle Tree has two main features:

\begin{enumerate}
    \item The root of the tree contains a hiding commitment to the entire collection $M$.
    \item The root can also serve as a commitment to each item $m \in M$. Furthermore, the proof that $m$ is a leaf of the tree reveals nothing about the other items in $M$ and has length 
    $O(\log q)$, where $q$ is the total number of leaves in the tree. 
\end{enumerate}

A Merkle tree can be constructed from a cryptographic hash function. Concretely, given a cryptographic hash function $H$ and a collection of items $m_0,...,m_{q-1}$, construct a binary tree with these items as the leaves. 

The leaves of the Merkle Tree are the zeroth level $h_{0,0}, h_{0,1}, ... , h_{0,q-1}$, where $h_{0,i} = m_i$. The next level has the same number of nodes $h_{1,0}, ..., h_{1,q-1}$ defined by $h_{1,i} = H(r_i || m_i)$ where $r_i$ is a random nonce. Level $k$ where $k \geq 2$, has half as many nodes as level $k-1$, defined by $h_{k,i} = H(h_{k-1,2i}||h_{k-1,2i+1})$. Figure~\ref{fig:merkletree} illustrates an example of a Merkle tree. In total there are $\ell_t+1$ levels where $\ell_q := \bigceil{\log_2 q}+1$. With this notation, $h_{\ell_q,0}$ is the value at the root of the Merkle Tree. 

The root of a Merkle tree $h_{\ell_q, 0}$ is a commitment to the entire collection due to collision resistance of $H$. To commit to the data $M$, the committer will generate $r_0,...,r_{q-1}$, compute the Merkle Tree, and announce $h_{\ell_q, 0}$. In the reveal step, the committer can announce $\bigbrace{(m_i,r_i)}_{i=0}^{q-1}$, and anyone can then compute the resulting Merkle tree and confirm that the root is equal to $h_{\ell_q, 0}$. 
\subsubsection{Merkle Proofs}

The root also serves as a commitment to each $m_i \in M$. Suppose someone who knows $m_i$ wants a proof that $m_i$ is a leaf in the Merkle tree. A proof $\pi(m_i)$ can be constructed from the Merkle tree. Furthermore, this proof reveals nothing about the other items $\bigbrace{m_j}_{j \neq i}$.

Define $x_0,x_1,...x_{\ell_q}$ recursively as:
\begin{align*}
    x_0 &:= i \\
    x_j &:= \bigfloor{ \frac{x_{j-1}}{2} } \text{ for } 1 \leq j \leq \ell_q. 
\end{align*}
With this notation, $\bigbrace{h_{j, x_j}}_{j=0}^{\ell_q}$ is the path from $m_i$ to the root of the Merkle Tree. The Merkle proof for $m_i$ is denoted as $\pi(m_i)$ and is given by
\begin{align*}
    \pi(m_i) &:= \bigbrace{r_i} \cup \bigbrace{\text{sibling}(h_{j,x_j})}_{j=1}^{\ell_q}, \text{ where } \\
    \text{sibling}(h_{i,j}) &:= 
    \casewise{
        \begin{tabular}{cc}
            $h_{i,j+1}$ & if $j$ is even, \\
            $h_{i,j-1}$ & if $j$ is odd. 
        \end{tabular}
    }
\end{align*}
See Supplementary Material~\smix{} for details on the binding and hiding properties of Merkle commitments, and how to verify the correctness of Merkle proofs. 

\begin{definition}[Merkle Commitment]
Given a data set $M = \bigbrace{m_1,...,m_t}$ and a set of random nonce values $r = \bigbrace{r_1,...,r_t}$, we use $\textsf{MCommit}(M,r)$ to denote the root of the Merkle Tree constructed from the data $M$ and random nonces $r$. 
\end{definition}

We refer the interested reader to Section 8.9 of \cite{BonehShoup20} for more details on Merkle Trees. 

\begin{figure}
    \centering
    \includegraphics[width = \textwidth]{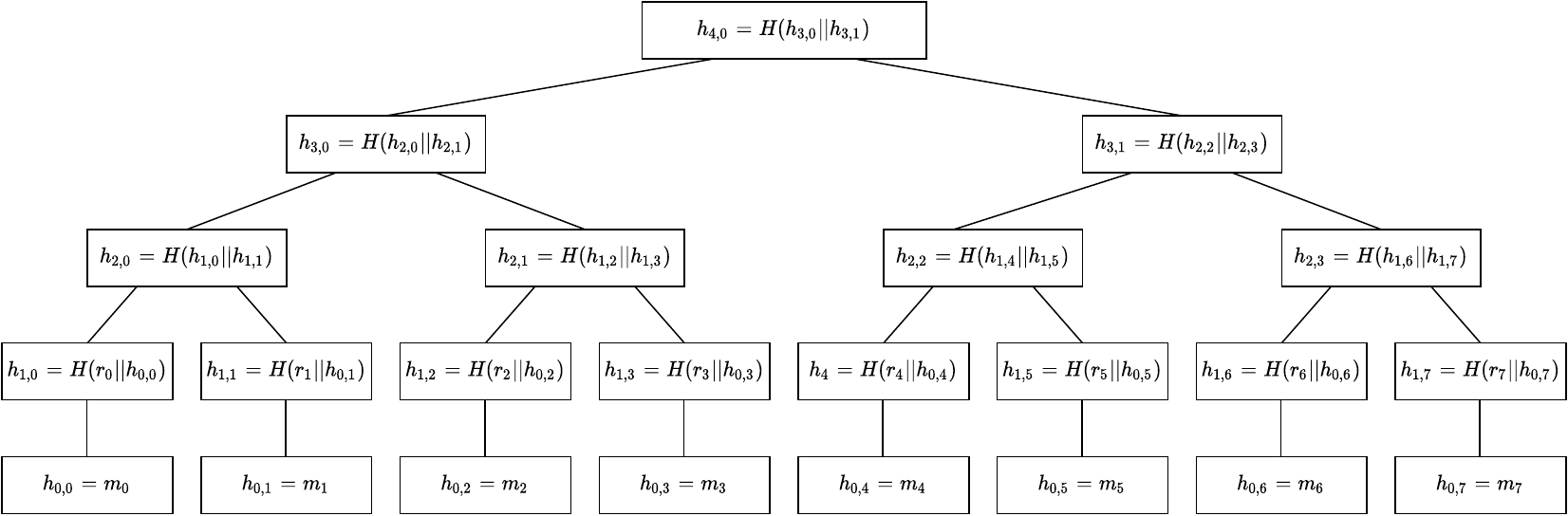}
    \caption{An example of a Merkle tree containing 8 items. Each item $m_i$ is a leaf node and has one parent which is $H(r_i||m_i)$, where $r_i$ is a random hiding nonce. All other internal nodes are computed by applying $H$ to the concatenation of its children.}
    \label{fig:merkletree}
\end{figure}

\subsubsection{Digital Signatures}

A digital signature scheme is comprised by three functions: $\textsf{Gen}, \textsf{sign}, \textsf{verify}$. $\textsf{Gen}()$ is a random function that produces valid public and private key pairs $(\textsf{pk}, \textsf{sk})$. Given a message, a signature is produced using the secret key via $\sigma = \textsf{sign}(\textsf{sk}, m)$. The authenticity of the signature is checked using the public key via $\textsf{verify}(\textsf{pk}, m, \sigma)$. A secure digital signature scheme has two properties:
\begin{enumerate}
    \item \textit{Correctness -} For a valid key pair $(\textsf{pk}, \textsf{sk})$ obtained from $\textsf{Gen}()$ and any message $m$, we have $\textsf{verify}(\textsf{pk}, m, \textsf{sign}(\textsf{sk},m)) = \texttt{True}$. 
    \item \textit{Secure -} Given a public key $\textsf{pk}$, if the corresponding secret key $\textsf{sk}$ is unknown, then it is computationally intractable to forge a signature on any message. Specifically, if $\textsf{sk}$ has never been used to sign a message $m'$, then without knowledge of $\textsf{sk}$, it is computationally intractable to find $(m',\sigma')$ so that $\textsf{verify}(\textsf{pk}, m', \sigma') = \texttt{True}$. 
\end{enumerate}

\noindent We refer the interested reader to Section 13 of \cite{BonehShoup20} for more details on digital signatures.

\subsubsection{Public Key Encryption}

A public key encryption scheme is specified by three functions: a key generation function, an encryption function $E$, and a decryption function $D$. In a public key encryption scheme, each user has a public key and private key denoted $(\textsf{pk},\textsf{sk})$, produced by the key generation function. As the name suggests, the public key $\textsf{pk}$ is known to everyone, while each secret key $\textsf{sk}$ is known only by its owner. Encryption is done using public keys, and decryption is done using secret keys. To send a message $m$ to Bob, one would encrypt $m$ using Bob's public key via $c = E(\textsf{pk}_{\text{Bob}}, m)$. Then Bob would decrypt the message via $D(\textsf{sk}_{\text{Bob}}, c)$. A secure Public Key Encryption scheme has two properties:

\begin{enumerate}
    \item \textit{Correctness -} For every valid key pair $(\textsf{pk}, \textsf{sk})$ and any message $m$, we have $m = D(\textsf{sk}, E(\textsf{pk},m))$, i.e., the intended recipient receives the correct message upon decryption.
    \item \textit{Secure -} For a public key $\textsf{pk}$ and any message $m$, if the corresponding secret key $\textsf{sk}$ is not known, then it is computationally intractable to deduce anything about $m$ from the ciphertext $E(\textsf{pk}, m)$.
\end{enumerate}

The appeal of public key encryption is that users do not have to have a shared common key in order to send encrypted messages to one another. We refer the interested reader to Part II of \cite{BonehShoup20} for more details on Public Key Encryption.

\subsubsection{Zero Knowledge Proofs}

A zero knowledge proof for a mathematical problem is a technique whereby one party (the prover) can convince another party (the verifier) that it knows a solution $w$ to the problem without revealing any information about $w$ other than the fact that it is a solution. Before discussing zero knowledge proofs further, we must first introduce proof systems.

\begin{definition}[Proof System]
Consider an arithmetic circuit $C : \mathcal{X} \times \mathcal{W} \rightarrow \bigbrace{0,1}$, and the following optimization problem: For a fixed $x \in \mathcal{X}$, find a $w \in \mathcal{W}$ so that $C(x,w) = 0$. Here $x$ is part of the problem statement, and $w$ is a solution candidate. Consider a tuple of functions $(S,V,P)$ where
\begin{enumerate}
    \item $S$ is a \textit{preprocessing function} that takes as input $C,x$ and outputs public parameters $\texttt{pp}$. 
    \item $P$ is a \textit{prover function} that takes as input $\texttt{pp},x,w$ and produces a proof $\pi$.
    \item $V$ is a \textit{verification function} that takes as input $\texttt{pp},x,\pi$ and outputs either $0$ or $1$ corresponding to whether the proof $\pi$ is invalid or valid respectively. 
\end{enumerate}
The tuple $(S,V,P)$ is a proof system for $C$ if it satisfies the following properties:
\begin{enumerate}
    \item \textit{Completeness -} If $C(x,w) = 0$, then $V(\texttt{pp}, x, P(\texttt{pp}, x, w))$ should evaluate to $1$; i.e., the verifier should accept proofs constructed from valid solutions $w$.
    \item \textit{Proof of Knowledge -} If $V(\texttt{pp},x,\pi) = 1$, then whoever constructed $\pi$ must have known a $w$ satisfying $C(x,w) = 0$.
\end{enumerate}
\end{definition}

With this definition in hand, we can now define zero knowledge proof systems. 

\begin{definition}[Zero Knowledge Proof Systems]
Consider a proof system $(S,V,P)$ for the problem of finding $w$ so that $C(x,w) = 0$. $(S,V,P)$ is a zero knowledge proof system if it is computationally intractable to learn anything about $w$ from $\pi := P(\texttt{pp}, x, w)$. If this is the case, then $\pi$ is a zero knowledge proof. 
\end{definition}

Zero knowledge proofs were first proposed by \cite{GoldwasserMR89}, but the prover and verifier functions were not optimized to be computationally efficient. In the next section, we present zk-SNARKs, which are computationally efficient zero knowledge proof systems. 

\subsubsection{zk-SNARKs}

In this section we introduce Succinct Non-interactive Arguments of Knowledge (SNARK). SNARKs are proof systems where proofs are short, and both the construction and verification of proofs are computationally efficient. 

\begin{definition}[Succinct Non-interactive Argument of Knowledge (SNARK)]
Consider the problem of finding $w \in \cW$ so that $C(x,w) = 0$, where $C$ is an arithmetic circuit with $n$ logic gates. A proof system $(S,V,P)$ is a SNARK if
\begin{enumerate}
    \item The runtime of the prover $P$ is $\widetilde{O}(n)$,
    \item The length of a proof computed by $P$ is $O(\log n)$,
    \item The runtime of the verifier $V$ is $O(\log n)$. 
\end{enumerate}
\end{definition}

\begin{definition}[zk-SNARK]
If a SNARK $(S,V,P)$ is also a zero knowledge proof system, then it is a zk-SNARK. 
\end{definition}

The Zcash cryptocurrency, which provides fully confidential transactions, was the first setting where zk-SNARKs have been used in the field \cite{SassonEtAl14}. zk-SNARKs have also been deployed in the zk-rollup procedure which increases the transaction throughput of the Ethereum blockchain \cite{Buterin16}. 

%PLONK \cite{GabizonWC19} is a recently developed zk-SNARK that we will use for PMM.  
For PMM we will need a zk-SNARK that does not require a trusted setup. PLONK \cite{GabizonWC19}, Sonic \cite{MallerBKM19}, and Marlin \cite{ChiesaHMMVW20} using a DARK based polynomial commitment scheme described in \cite{BunzFS20, BlockHRRS21}. Other options include Bulletproofs \cite{BunzBBPWM18} and Spartan \cite{Setty20}.

\subsection{Implementation Details and Examples}\label{app:examples}

In this section, we show how driver period information in ridehailing services, and a mobility provider's impact on congestion can be obtained from the protocol. Both cases involve specifying characteristics of the query function $g$ and trip metadata that enable the desired information to be computed by the protocol. 

\subsubsection{Obtaining ridehailing period activity}\label{app:phase}

As discussed in Example~\ref{regulation:period}, the pay rate of ridehailing drivers depends on the period they are in. Ridehailing companies use period 2 to tell users that they are matched and a driver is en route, thereby reducing the likelihood that the user leaves the system out of impatience. Due to this utility, period 2 has a higher pay rate than period 1. There is thus a financial incentive for ridehailing companies to report period 2 activity as period 1 activity so that they can have improved user retention while keeping operations costs low. Accurate period information is thus important to protect the wages of ridehailing drivers. 

We achieve accurate period information by including digital signatures in the trip metadata. Recall that the trip metadata includes the request time, match time, pickup time, and dropoff time of the request. The period 2 and period 3 activity associated with a trip can be deduced from these timestamps, as shown in Figure~\ref{fig:timestamps}. Furthermore, Rider Witness and ARA ensure that reporting the true demand is a dominant strategy for the ridehailing operator. 

\begin{figure}
    \centering
    \includegraphics[width = \textwidth]{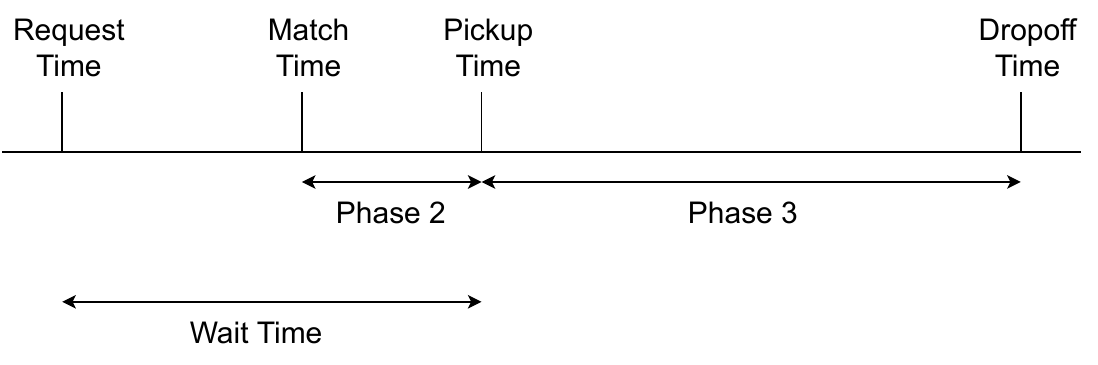}
    \caption{The timesteps within the trip metadata determine the Period 2 and Period 3 activity of the vehicle that serves this trip.}
    \label{fig:timestamps}
\end{figure}

Therefore to ensure accurate period information, it is sufficient to ensure that the aforementioned timesteps are recorded correctly. For period 2 accuracy, we need to ensure that the match time and pickup time are recorded properly for each trip. To do this, we will use digital signatures. To notify a user that they have been matched, the ridehailing operator will send $(m_{pt}, \sigma_{pt})$, where:
\begin{align*}
    m_{pt} &= \text{You have been matched to vehicle }\texttt{vehID}\text{ at time }\texttt{currtime}, \\
    \sigma_{pt} &= \textsf{sign}(\textsf{sk}_{\text{mp}}, m_{pt}).
\end{align*}
The user will only consider the message $m_{pt}$ as genuine if it is accompanied by a valid signature $\sigma_{pt}$. Therefore, telling a user they are matched (and thus reducing the likelihood that this user cancels their trip) requires the ridehailing company to provide an irrefutable and unforgeable declaration of the match time in the form of $(m_{pt}, \sigma_{pt})$. The message and signature $(m_{pt}, \sigma_{pt})$ is then included in the trip metadata to certify the trip's match time. The same can be done for the pickup time, and as a result, ensure accurate reporting of all period 2 activity. 

The accuracy of period 3 activity can be ensured by ensuring that pickup time and dropoff time are recorded correctly. 

To implement driver wage inspection through the protocol, the query function $g$ would be
\begin{align*}
    g_{\text{wage}}(\Lambda) := \prod_{\lambda \in \Lambda} \mathds{1}[w(\lambda) = f_{\text{wage}}(\lambda)],
\end{align*}
where $w(\lambda)$ is the driver wage of ride $\lambda$, and $f_{\text{wage}}$ is the MP's wage formula which may depend on the period and trajectory information contained in the trip metadata of $\lambda$. Note that $g_{\text{wage}}(\Lambda) = 1$ if and only if all drivers were paid properly, and is $0$ otherwise.  

\begin{remark}[Evaluating Waiting Time Equity]
Using the idea from Section~\ref{app:phase}, one can also evaluate the equity of waiting times throughout the network. It is clear from Figure~\ref{fig:timestamps} that the wait time can be determined by the request time and pickup time, both of which can be found in the trip metadata. The trip metadata also includes the pickup location and dropoff location, so the average wait time as a function of pickup location, dropoff location, both pickup and dropoff locations, can all be computed from the trip metadata. 

To implement a waiting time evaluation through the protocol, the Municipal Authority would specify a fairness threshold $\tau$. The query function $g$ is then designed to output $1$ if and only if the average waiting time across locations does not vary by more than the pre-specified threshold, and outputs $0$ otherwise. Concretely, if we want to enforce wait time equity across pickup regions, we could do this with the function

\begin{align*}
    g_{\text{wait}}(\Lambda) = \prod_{i,j \in V} \mathds{1} \bigbra{ \abs{\tau_i - \tau_j} \leq \tau }
\end{align*}
where $\tau_i$ is the average wait time for requests in region $i$. 
\end{remark}

\subsubsection{Evaluating contributions to congestion}

The trip metadata contains the trip trajectory which can be used to evaluate a ridehailing fleet's contribution to congestion. The trip trajectory provides the location of the service vehicle as a function of time, which provides two important insights. First, the trip trajectories can be used to determine how many ridehailing vehicles are on a particular road at any given time. Second, from a trajectory one can compute the amount of time the vehicle spends on each road within the trip path. Thus the average travel time for a road can be calculated, which can then be used to estimate the total traffic flow on the road using traffic models. Combining these two pieces of information, the fraction of a road's total traffic that is ridehailing vehicles can be computed from the trip metadata. 

\subsection{Necessity of Assumption~\ref{assump:honest} for Verifiability}\label{app:assump:honest:justification}
In this section we show that Assumption~\ref{assump:honest} is necessary for verifiable queries on mobility data under the natural assumption that MA does not have surveillance in the interior of MP vehicles. This assumption on limited surveillance ability of MA is formalized in Assumption~\ref{assump:surveil_limit}.
\begin{assumption}\label{assump:surveil_limit}
There does not exist a practical way for MA to determine whether a MP vehicle is carrying a customer or not, without directly tracking all customers. In particular, MA cannot determine the period information of MP vehicles. 
\end{assumption}
Note that MA can obtain phase information from the drivers or from MP, but in the absence of Assumption 1, drivers and MPs may act strategically, and may not be trustworthy. 

The following result shows that under Assumption~\ref{assump:surveil_limit}, if the drivers or riders are willing to collude with the MP, then the MP can misreport properties of its mobility demand in a way that is undetectable by the MA.

\begin{observation}[Necessity of Assumption~\ref{assump:honest} for Strategyproofness]\label{obv:nec_honest}
Under Assumption~\ref{assump:surveil_limit}, the following events are undetectable by MA, even if MA can track all MP vehicles (i.e., knows the location of each MP vehicle at any time):
\begin{enumerate}
    \item If drivers collude with MP, then MP can overreport demand.
    \item If riders and drivers collude with MP, then MP can underreport demand, or misreport attributes of the demand. 
\end{enumerate}
\end{observation}

\begin{proof}[Proof of Observation~\ref{obv:nec_honest}]
By Assumption~\ref{assump:surveil_limit}, MA cannot distinguish between MP vehicles in period 1 and MP vehicles in period 2 or 3. Suppose the drivers are willing to collude with MP. If MP wants to overreport demand from some origin $i$ to some destination $j$, it can have some drivers drive from $i$ to $j$ without a passenger. This will lead to period 1 traffic from $i$ to $j$, however the drivers will report themselves in period 3 to MA. This way, even if MA is able to track the MP vehicles, the reported period information from the drivers will be consistent with the demand report from MP. 

Now suppose both riders and drivers are willing to collude with MP. If the MP wants to underreport demand from $i$ to $j$, they can have some drivers who are serving passengers from $i$ to $j$ report themselves in period 1 to MA. 
\end{proof}

So in the absence of Assumption~\ref{assump:honest}, the MA has no way of checking whether the messages it receives from MP are computed from the true demand. 

\begin{remark}[Tracking Users is also insufficient]
Even if the MA is able to track users and thus determine whether a MP vehicle has a passenger, this still does not prevent overreporting of demand. In this case, MP can hire people to hail rides from specific trips if it wants to overreport demand. 
\end{remark}

\subsection{Roadside Audits with fewer sensors}\label{app:RRA}

In this section we present the Randomized Roadside Audits (RRA) mechanism. Like ARA, the RRA detects overreporting of demand by conducting road audits. Where ARA places sensors on every road, RRA places sensors on a small subset of randomly selected roads, enabling it to use fewer sensors. 

The sensors used in RRA are similar to the sensors used in ARA described in Section~\ref{sec:protocol:mech:ARA:sensors}, with the following differences:

\begin{enumerate}
    \item Each sensor has its own pair of public and secret keys $(\textsf{pk}_s, \textsf{sk}_s)$ for digital signatures. Everyone knows $\textsf{pk}_s$, but $\textsf{sk}_s$ is contained in a Hardware Security Module within the sensor so that it is impossible to extract $\textsf{sk}_s$ from the sensor, but it is still possible to sign messages using $\textsf{sk}_s$. 
    \item Each sensor now records its own location using GPS. 
\end{enumerate}

First, the MA and MP agree on a list of public keys belonging to the sensors. In particular, they must agree on the number of sensors being deployed in the network. Let $mp$ be the number of sensors being deployed (recall that $m$ is the number of roads in the network). We will focus on the case where $p \in (0,1)$. If $p > 1$, then there are enough sensors to implement ARA.

During the data-collection period, the MA will place the sensors inside vehicles which are driven by its employees. We assume that MP cannot determine which vehicles are carrying sensors. In practice, MA can have much more than $mp$ employees driving around in the network, but only $mp$ of them will have sensors. 

The data collection period is divided up into many rounds (e.g., a round could be 1 hour long). In each round, MA will sample a random set of $mp$ roads. Each vehicle with a sensor is assigned to one of these roads, where they will stay (i.e., parked on the side of the street) to measure the MP traffic that pass by them. The sensor will record a measurement $u$ which specifies the time, the period and the location of every MP vehicle that passes by. It will then sign the message with its secret key via $\sigma_u := \textsf{sign}(\textsf{sk}_s, u)$. In particular, a sensor assigned to road $e \in E$ in round $t$ will be able to determine $\varphi_t(e,\Lambda)$, which is the total number of period 2 or period 3 MP vehicles that traverse $e$ in round $t$, formally described below. 
\begin{align*}
    \varphi_t(e,\Lambda) := \sum_{\lambda \in \Lambda} \mathds{1}_{[\lambda \text{ traverses }e \text{ in round }t]}.
\end{align*}

As was the case in ARA, the sensors have a communication constraint that prevents them from transmitting their data unless both MA and MP give permission. Therefore during the data collection period, MP does not know where the sensors are. Once the data collection period is over, both MA and MP give permission to collect the data from the sensors. 

\begin{definition}[RRA Test]
The RRA test checks whether the road usage on sampled roads is consistent with the demand reported by MP. Concretely, a witness $w = (\Lambda_w,r_w,c_w)$ passes the RRA test if $\varphi_t(e,\Lambda) = \varphi_t(e,\Lambda_w)$ for all pairs $(e,t)$ such that $e$ was sampled in round $t$.
\end{definition}

\begin{observation}[Efficacy of Random Roadside Audits]\label{obv:RRA}
Under Assumption~\ref{assump:honest}, if MP submits a commitment $\sigma' = \textsf{MCommit}(\Lambda',r)$ to a strict superset of the demand, i.e., $\Lambda \subset \Lambda'$, then with probability at least $p$, any proof submitted by MP will either be inconsistent with $\sigma'$ or will fail the RRA test. Hence overreporting of demand will be detected with positive probability. 
\end{observation}

\begin{proof}[Proof of Observation~\ref{obv:RRA}]

We use a similar analysis to ARA. Suppose MP overreports the demand, i.e., submits a commitment $\sigma' = \textsf{MCommit}(\Lambda',r)$ where $\Lambda'$ is a strict superset of $\Lambda$. Then there exists $\lambda' \in \Lambda' \setminus \Lambda$. Let $e'$ be any road in the trip trajectory of $\lambda'$. and $t(\lambda',e')$ be the round in which trip $\lambda'$ traverses $e'$. We then have $\varphi_{t(\lambda',e')}(e',\Lambda) < \varphi_{t(\lambda',e')}(e',\Lambda')$. If $e'$ is audited in round $t(\lambda',e')$, then $\sigma'$ will be inconsistent with the roadside audit measurements, and will fail the RRA test. Since MA samples $mp$ roads to audit uniformly randomly in each round, and there are a total of $m$ roads, the probability that $e'$ is chosen in round $t(\lambda',e')$ is $p$. Since overreporting is detected only probabilistically, in the event that it is detected, MA should fine MP so that MP's expected utility is reduced if it overreports demand.  

\end{proof}

\begin{remark}[Comparing RRA to ARA]
When compared to ARA, RRA uses fewer sensors. This, however, is not without drawbacks, since RRA detects demand overreporting only probabilistically. Thus in RRA the MA needs to fine the MP in the event that demand overreporting is detected. In particular, the fine should be chosen so that the MP's expected utility is decreased if it decides to overreport demand. Concretely, suppose $U_h,U_d$ are the utilities received by MP when acting honestly and dishonestly respectively. Since dishonesty is detected with probability $p$, the fine $F$ must satisfy
\begin{align*}
    U_h > (1-p) U_d - p F \implies F > \frac{1}{p}(U_d - U_h) - U_d. 
\end{align*}
If MA is using very few sensors or if $U_d$ is much larger than $U_h$, then $F$ needs to be very large. A large fine, however can be difficult to implement. Recall from Section~\ref{sec:protocol:mech:ARA} that inconsistencies between demand metadata and roadside measurements  due to GPS errors can occur even if all parties are honest. If such errors occur, then an MP would incur a large fine even if it behaves honestly. For this reason, even an honest MP may not want to participate in the protocol. One could use an error tolerant version of RRA, but for large $F$ the tolerance parameter $\epsilon$ would need to be large, enabling a dishonest MP to overreport demand while remaining within the tolerance parameter. 
\end{remark}

\subsubsection{Security Discussion}

We now make several remarks regarding the two sensor modifications we made for RRA. First, the signatures generated by the sensors' secret keys ensure that MA cannot fabricate or otherwise tamper with the sensor's data. This is important because the sensors are in the possession of the MA and its employees. Even if the MA manages to change the data in the sensor's storage, it cannot produce the corresponding signatures for the altered data since it does not know the secret key, which is protected by a Hardware Security Module. 

Second, the sensor's location data is essential to prevent MA from conducting relay attacks. A relay attack is as follows: Suppose Alice and Bob are both MA employees. Alice has a sensor in her car. Bob does not have a sensor in his car, but he wants to collect data as if he had a sensor in his car. The MA can give Bob an unofficial sensor (this sensor does not have a valid public and secret key recognized by the MP), allowing Bob to detect signals from MP vehicles. Since Bob's sensor does not have an official secret key, he cannot obtain a valid signature for his measurements. To get the signatures, Bob sends the detected signal to Alice, and Alice relays the signal to her sensor, which will sign the measurement and record it. In this manner, the MA is able to get official measurements and signatures on Bob's road even though he does not have an official sensor. Fortunately, this attack is thwarted if the sensor knows its own location. If a sensor receives a measurement whose location is very different than its location, then it will reject the message, thus thwarting the relay attack.\footnote{The measurements can be protected by authenticated encryption so that relayers (e.g., Bob) cannot modify the messages (i.e., changing the vehicle position part of the measurement)} 

\subsection{Establishing Verifiability and Differential Privacy for Appendix~\ref{sec:ext:differentialprivacy}}\label{app:diffpriv_efficacy}

Verifiability is established by steps 1, 3 and 4 of $C$. Based on the analysis in Section~\ref{sec:protocol:stratproof}, a witness satisfies step 1 of $C$ if and only if $\Lambda_w = \Lambda$, i.e., the demand is reported honestly. Since the Pedersen commitment scheme is secure, it is computationally binding, meaning that it is computationally intractable for MP to find $Z_{\text{mp}}', r_{\text{mp}}'$ with $Z_{\text{mp}} \neq Z_{\text{mp}}'$ and $\textsf{Commit}(Z_{\text{mp}}', r_{\text{mp}}') = z_{\text{mp}}$. So in order for the MP's witness to pass step 3 of $C$, it must have $Z_{\text{mp},w} = Z_{\text{mp}}$. Given steps 1, 3 have passed, step 4 ensures that the message $z$ is indeed equal to $\widetilde{g}(\Lambda, Z_{\text{ma}} \oplus Z_{\text{mp}})$, which establishes verifiability. 

To establish differential privacy, we need to show two things: (a) MA does not know $Z_{\text{ma}} \oplus Z_{\text{mp}}$ (see Remark~\ref{rem:MA_randomness}) and (b) $Z_{\text{ma}} \oplus Z_{\text{mp}}$ is uniformly distributed over $\bigbrace{0,1}^v$, even if MA and MP are acting strategically. To this end, we consider a game between MA and MP with actions $Z_{\text{ma}},Z_{\text{mp}} \in \bigbrace{0,1}^v$ and outcome $Z_{\text{ma}} \oplus Z_{\text{mp}} \in \bigbrace{0,1}^v$. We will show that the strategy profile where both $Z_{\text{ma}}, Z_{\text{mp}}$ are independently sampled uniformly at random is a Nash equilibrium, meaning that differential privacy is achieved as long as at least one party is honest. 

To show that independent uniform random sampling of both $Z_{\text{ma}}, Z_{\text{mp}}$ is a Nash equilibrium, we first need to show that $Z_{\text{ma}},Z_{\text{mp}}$ are independent. In the protocol $Z_{\text{mp}}$ is sampled first, and a Pedersen commitment $z_{\text{mp}}$ is sent to MA. Since Pedersen commitments are perfectly hiding, the distribution of $z_{\text{mp}}$ does not depend on $Z_{\text{mp}}$. So even if MA samples $Z_{\text{ma}}$ based on the value of $z_{\text{mp}}$, the result will be independent of $Z_{\text{mp}}$. Now that we have established independence of $Z_{\text{mp}}, Z_{\text{ma}}$, we make use of the following observation. 

\begin{observation}[One Time Pad]\label{obv:otp}
Suppose $Z_{\text{ma}},Z_{\text{mp}}$ are independent random variables. If $Z_{\text{ma}}$ is uniformly distributed over $\bigbrace{0,1}^v$, then $Z_{\text{ma}} \oplus Z_{\text{mp}}$ is uniformly distributed over $\bigbrace{0,1}^v$, regardless of how $Z_{\text{mp}}$ is sampled. If $Z_{\text{mp}}$ is uniformly distributed over $\bigbrace{0,1}^v$, then $Z_{\text{ma}} \oplus Z_{\text{mp}}$ is uniformly distributed over $\bigbrace{0,1}^v$, regardless of how $Z_{\text{ma}}$ is sampled. 
\end{observation}
Observation~\ref{obv:otp} says that if $Z_{\text{ma}},Z_{\text{mp}}$ are independent, the distribution of $Z_{\text{ma}} \oplus Z_{\text{mp}}$ does not depend on $Z_{\text{mp}}$ if $Z_{\text{ma}}$ is uniformly random, and vice versa. Hence independent uniform sampling of $Z_{\text{ma}},Z_{\text{mp}}$ is a Nash equilibrium, establishing condition (b). To establish (a), if at least one party is honest, then we can assume without loss of generality that both parties are acting according to the Nash equilibrium. By observation~\ref{obv:otp}, this means the marginal distribution of $Z_{\text{ma}} \oplus Z_{\text{mp}}$ and the conditional distribution of $Z_{\text{ma}} \oplus Z_{\text{mp}}$ given $Z_{\text{ma}}$ are both uniform. In particular, MA does not learn anything about $Z_{\text{ma}} \oplus Z_{\text{mp}}$ from $Z_{\text{ma}}$.

\subsection{More Details on Congestion Pricing}\label{app:cong_price}

When the travel cost is the same as travel time, the prices can be obtained from the following optimization problem
\begin{align*}
    \min & \; \sum_{e \in E} x_e f_e(x_e) \\
    \text{s.t. } & x = \sum_{o \in V} \sum_{d \in V} x^{od} \nonumber \\
    & x^{od} \succeq 0 \; \forall o \in V, d \in V \nonumber \\
    & \sum_{(u,v) \in E} x^{od}_{(u,v)} - x^{od}_{(v,u)} = \Lambda(o,d) \bigpar{ \mathds{1}_{[u = o]} - \mathds{1}_{[u = d]} } \forall u \in V \nonumber
\end{align*}
where $x_e^{od}$ is the traffic flow from $o$ to $d$ that uses edge $e$, and $x$ is the total traffic flow. $\Lambda$ is the travel demand where $\Lambda(o,d)$ is the rate at which users require transport from $o$ to $d$. Here the objective measures the sum of the travel times of all requests in $\Lambda$.

Let $x^*$ be a solution to \eqref{opt:so_tt}. By first order optimality conditions\footnote{i.e., it should be impossible to decrease the objective function by reallocating flow from $p_1$ to $p_2$ or vice versa.} of $x^*$, for any origin-destination pair $(i,j)$, and any two paths $p_1,p_2$ from $i$ to $j$ with non-zero flow, we have
\begin{align}
    \sum_{e \in p_1} \evaluate{\frac{\partial}{\partial x_e} x_e f_e(x_e)}_{x_e = x_e^*} &= \sum_{e' \in p_2} \evaluate{\frac{\partial}{\partial x_{e'}} x_{e'} f_{e'}(x_{e'})}_{x_{e'} = x_{e'}^*} \nonumber \\
    \implies \sum_{e \in p_1} f_e(x_e^*) + x_e^* f_e'(x_e^*) &= \sum_{e' \in p_2} f_{e'}(x_{e'}^*) + x_{e'}^* f_{e'}'(x_{e'}^*). \label{eqn:tt:KKT}
\end{align}
In order to realize $x^*$ as a user equilibrium, the costs of $p_1,p_2$ should be the same so that no user has an incentive to change their strategy. This can be achieved by setting the toll for each road $e$ as $p_e := x_e^* f_e'(x_e^*)$. By doing so, from \eqref{eqn:tt:KKT} we can see that the cost (travel time plus toll) for the two paths will be equal. 

In the context of PMM, the function $g$ associated with congestion pricing is
\begin{align*}
    g_{\text{cp}}(\Lambda) := \bigbrace{ x_e^* f_e'(x_e^*) }_{e \in E} \text{ where } x^* \text{ solves } \eqref{opt:so_tt}.
\end{align*}

\subsection{Efficacy of Merkle Proofs}\label{app:merkle_proof}

To verify the proof $\pi(m_i) = \bigbrace{r_i} \cup \bigbrace{\text{sibling}(h_{j,x_j})}_{j=1}^{\ell_q}$ for membership of $m_i$, the recipient of the proof would compute $v_1,v_2,...,v_{\ell_q-1}$ recursively via:
\begin{align*}
    v_1 &:= H(r_i||m_i) \\
    v_{j} &:= \casewise{
        \begin{tabular}{cc}
            $H(v_{j-1}||\text{sibling}(h_{j-1,x_{j-1}}))$ & if $x_{j-1}$ is even, \\
            $H(\text{sibling}(h_{j-1,x_{j-1}})||v_{j-1})$ & if $x_{j-1}$ is odd,
        \end{tabular}
    } \text{ for } 1 \leq j < \ell_q.
\end{align*}
By the construction of the Merkle tree, $v_j = h_{j,x_j}$, and so in particular the Merkle Proof is valid if and only if $v_{\ell_q}$ is equal to the root, i.e., $v_{\ell_q} = h_{\ell_q,0}$.

Since there are $q$ leaves in the binary tree, $p_i$ has at most $\log_2 q$ vertices in it, and each hash is $d$ bits, so the length of $\pi$ is at most $d \log_2 q$. 

By collision resistance of $H$, it is intractable to forge a proof if $m_i$ is not in the tree, and since hiding nonces are used when hashing the items, the proof reveals nothing about the other items in the tree. 

\fi 

\end{document}

%% file: peripherals/macros.tex
\newcommand{\abs}[1]{\left| #1 \right|}
\newcommand{\bigpar}[1]{\left( #1 \right)}
\newcommand{\bigbra}[1]{\left[ #1 \right]}
\newcommand{\bigbrace}[1]{\left\{ #1 \right\}}

\newcommand{\bigceil}[1]{\left\lceil #1 \right\rceil}
\newcommand{\bigfloor}[1]{\left\lfloor #1 \right\rfloor}
\newcommand{\casewise}[1]{\left\{ #1 \right.}
\newcommand{\evaluate}[1]{\left. #1 \right|}

\newcommand{\cD}{\mathcal{D}}

\newcommand{\cF}{\mathcal{F}}

\newcommand{\cM}{\mathcal{M}}

\newcommand{\cR}{\mathcal{R}}

\newcommand{\cT}{\mathcal{T}}

\newcommand{\cW}{\mathcal{W}}
\newcommand{\cX}{\mathcal{X}}